\documentclass[11pt]{article}

\usepackage[latin1]{inputenc}
\usepackage{amsthm}
\usepackage{amsmath}
\usepackage{amsfonts}
\usepackage{amssymb}
\usepackage[scale=0.8]{geometry}
\usepackage{hyperref}
\usepackage{bbm}
\usepackage{float}
\usepackage{caption}
\usepackage{graphicx}
\hypersetup{colorlinks=true,citecolor=blue, linkcolor=blue,
urlcolor=blue}

 \newcommand{\TreeD}{\mathbb{T}_\Delta}
  \newcommand{\Treed}{\mathbb{T}_d}
\newcommand{\IS}{\#{\sc IS}}
\newcommand{\SAT}{\#{\sc Sat}}
\newcommand{\AP}{$\leq_{\texttt {AP}}$}
\newcommand{\APequiv}{$\equiv_{\texttt {AP}}$}
\newcommand{\BIS}{\#{\rm BIS}}

\newcommand{\BMNSpin}[3]{{\sc Bi-M-Nonuniform-2-Spin}$(#1,#2,#3)$}
\newcommand{\BSpin}[3]{{\sc Bi-2-Spin}$(#1,#2,#3)$}
\newcommand{\BBSpin}[4]{{\sc Bi-2-Spin}$(#1,#2,#3,#4)$}

\newcommand{\inftree}[1]{\mathbb{T}_{#1}}

\newcommand{\ceil}[1]{\lceil #1\rceil}

\def\qm{{q^-}}
\def\qp{{q^+}}

\let\epsilon=\varepsilon

\def\Pr{\mathop{\rm Pr}\nolimits}
\def\calG{\mathcal{G}}
\def\phase{\pi}

\def\phases{\{-,+\}}

\newcommand{\eps}{\epsilon}

\def\pconfig{\tilde \sigma}

\newcommand{\Gc}{\ensuremath{\mathcal{G}}}
\newcommand{\E}{\ensuremath{\mathbf{E}}}

\newtheorem{theorem}{Theorem}
\newtheorem{lemma}[theorem]{Lemma}

\newtheorem{corollary}[theorem]{Corollary}

\newtheorem{definition}[theorem]{Definition}

\newtheorem{condition}{Condition}

\begin{document}

\title{\BIS-Hardness for 2-Spin Systems on Bipartite Bounded Degree Graphs in the Tree Non-Uniqueness Region}
\author{
  Jin-Yi Cai\thanks{
  Department of Computer Sciences, University of Wisconsin-Madison, USA. Research supported by NSF grant CCF-1217549.
  Heng Guo is supported by a 2013 Simons award for graduate students in theoretical computer science.}
\and
Andreas Galanis\thanks{
  Department of Computer Science, University of Oxford, Wolfson Building, Parks Road, Oxford, OX1~3QD, UK.
  The research leading to these results has received funding from the European Research Council under
  the European Union's Seventh Framework Programme (FP7/2007-2013) ERC grant agreement no.\ 334828. The paper
  reflects only the authors' views and not the views of the ERC or the European Commission.
  The European Union is not liable for any use that may be made of the information contained therein.}
  \and
  Leslie Ann Goldberg$^\dag$
  \and
  Heng Guo$^\ast$
  \and  Mark Jerrum\thanks{
  School of Mathematical Sciences,
  Queen Mary, University of London, Mile End Road, London E1 4NS, United Kingdom.}
\and
 Daniel \v{S}tefankovi\v{c}\thanks{
Department of Computer Science, University of Rochester,
Rochester, NY 14627.  Research
supported by NSF grant CCF-0910415.}
 \and Eric Vigoda\thanks{School of Computer Science, Georgia
Institute of Technology, Atlanta GA 30332.
 Research supported by
NSF grant CCF-1217458.}
 }

\date{}
\maketitle

\begin{abstract}
Counting independent sets on bipartite graphs (\BIS) is considered a canonical
counting problem of intermediate approximation complexity. It is conjectured 
that \BIS~neither has an FPRAS nor is as hard as \SAT~to approximate. We study
\BIS~in the general framework of two-state spin systems on bipartite graphs. We
define two notions, nearly-independent phase-correlated spins and unary symmetry
breaking. We prove that it is \BIS-hard to approximate the partition function
of any 2-spin system on bipartite graphs supporting these two notions. As a
consequence, we classify the complexity of approximating the partition function of
antiferromagnetic 2-spin systems on bounded-degree bipartite graphs. 
\end{abstract}

\section{Introduction}
There has been great progress in classifying the complexity of
counting problems recently.
One important success is for counting constraint satisfaction problems (\#CSP),
where a sweeping complexity dichotomy is proved~\cite{Bul08,DR13,CC12}.
While the landscape of exact counting becomes clearer,
the complexity of approximate counting remains mysterious.
Two typical classes of problems have been identified:
1) those that  have a fully polynomial-time randomized approximation scheme (FPRAS),
and 2) those that are \SAT-hard with respect to approximation preserving reductions (AP-reductions) \cite{DGGJ03}.
If $\text{NP}\neq \text{RP}$ then
\SAT~admits no FPRAS\footnote{In fact, Zuckerman proves a stronger result --- there is no FPRAS for the logarithm of the number of satisfying assignments unless NP=RP.}~\cite{Zuc96}, 
 and therefore 
neither does any  problem
in the second class.
These two classes are analogous to P-time
tractable vs.~NP-hard decision or optimization problems.

Interestingly, in approximate counting,
there has emerged a third distinct class of natural problems, which seems to be
of intermediate complexity.
It is conjectured~\cite{DGGJ03} that the problems in this class
do not have an FPRAS but that they are not as hard as \SAT\ to approximate.
A canonical problem
in this class has been identified, which is
to count the number of independent sets  in a bipartite graph (\BIS).
Despite many attempts, nobody has found an FPRAS for
\BIS\ or an AP-reduction from \SAT\ to \BIS.
The conjecture is that
neither exists. 
Mossel et al.~\cite{MWW09} showed that
the Gibbs sampler for sampling independent sets in bipartite graphs mixes  slowly
even on bipartite graphs of degree at most~$6$.
Another interesting attempted Markov Chain
for \BIS\
by Ge and \v{S}tefankovi\v{c} \cite{GS10}
was
also shown later to be
slowly mixing by Goldberg and Jerrum \cite{GJ12a}.

\BIS~plays an important role in classifying counting problems with respect to approximation.
A trichotomy theorem is shown for
the 
complexity
of approximately solving unweighted Boolean counting CSPs,
where in addition to
problems
that are solvable by FPRASes and those that are AP-reducible from \SAT,
there is the intermediate class
of problems which are
equivalent to \BIS~\cite{DGJ10}.
Many 
counting
problems are shown to be \BIS-hard and hence are conjectured to  have no FPRAS \cite{BDGJM13,CDGJLMR13},
including
estimating the partition function of the
the ferromagnetic Potts model~\cite{GJ12}.
Moreover, under AP-reductions \BIS~is complete in a logically defined class of
problems, called  \#RH$\Pi_1$, to which an increasing variety
of problems have been shown to belong.
Other typical complete problems in \#RH$\Pi_1$ include counting the
number of downsets in a partially ordered set \cite{DGGJ03}
and computing the partition function of the ferromagnetic Ising model with local external fields \cite{GJ07}.

The problem of counting independent sets (\IS) can be viewed as a special case in the general framework of spin systems,
which originated from statistical physics to model interactions between neighbors on graphs.
In this paper, we focus on two-state spin systems.
In general such a system is parameterized by edge weights $\beta,\gamma\geq 0$ and a vertex weight $\lambda>0$.
An instance is a graph $G=(V,E)$.
A configuration $\sigma$ is a mapping $\sigma : V\to \{0,1\}$ from vertices to (two) spins.
The weight $w(\sigma)$ of a configuration $\sigma$ is given by
\begin{equation}
\label{ourdefw}
w(\sigma) = \beta^{m_0(\sigma)} \gamma ^{m_1(\sigma)} \lambda ^{n_1(\sigma)}
\end{equation}
where $m_0(\sigma)$ is the number of $(0,0)$ edges given by the configuration $\sigma$,
$m_1(\sigma)$ is the number of $(1,1)$ edges,
and $n_1(\sigma)$ is the number of vertices assigned $1$.
We are interested in computing the partition function,
which is defined by
\begin{equation}\label{eq:pfdef}
Z_G(\beta,\gamma,\lambda) = \sum_{\sigma:V\rightarrow\{0,1\}}w(\sigma).
\end{equation}
The partition function is the normalizing factor of the Gibbs distribution,
which is
the distribution in which a configuration~$\sigma$ is drawn
with probability $\Pr_{G;\beta,\gamma,\lambda}(\sigma)=\frac{w(\sigma)}{Z_G(\beta,\gamma,\lambda)}$.
The spin system is called
\emph{ferromagnetic} if $\beta\gamma>1$ and \emph{antiferromagnetic} if $\beta\gamma<1$.
In particular, when $\beta=\gamma$, such a system is the famous \emph{Ising model},
and when $\beta=1$ and $\gamma=0$, it is the \emph{hard-core gas model},
the partition function
of which counts independent sets
when $\lambda=1$.
The external field $\lambda$ is typically referred to as the activity or fugacity of the hard-core model.

Approximating the partition function of the hard-core model is especially well studied.
We now know that the complexity transition from easy to hard corresponds exactly
to the uniqueness of the Gibbs measure in infinite $(\Delta-1)$-ary trees $\TreeD$ (for details of these notions, see \cite{Geo11}).
Notice that $(\Delta-1)$-ary trees are graphs of maximum degree $\Delta$, hence our use of the notation $\TreeD$.
On the algorithmic side, Weitz presented a fully polynomial-time approximation scheme (FPTAS) for the hard-core gas model
on graphs of maximum degree $\Delta$ when
uniqueness holds \cite{Wei06}.
On the other hand, Sly showed that
the approximation problem
is \SAT-hard for a small interval beyond the uniqueness threshold \cite{Sly10}.
Building on their work, it is now established that for all antiferromagnetic 2-spin systems
there is an FPTAS for graphs of maximum degree $\Delta$
up to the uniqueness threshold \cite{LLY13} (see also \cite{LLY12,SST12}),
whereas non-uniqueness implies \SAT-hardness under AP-reductions on $\Delta$-regular graphs \cite{SS12}
(see also \cite{CCGL12,GSV12}).
The only place that remains unclear is exactly at the uniqueness threshold.

A key feature of spin systems in the antiferromagnetic non-uniqueness region is
the ability to support a gadget with many vertices whose spins are highly correlated
with the phase of the gadget (which is either $+$ or $-$),
but are nearly independent among themselves conditioned on that phase.
Such a gadget was used by Sly~\cite{Sly10} to show inapproximability of the partition function of the hard-core model
when $\lambda$ is just above the uniqueness threshold.
A different gadget with similar properties was used by
Goldberg et al.~\cite{GJM12} to show inapproximability
on a planar graph when $\lambda$ is much larger.
We abstract this notion of nearly-independent phase-correlated spins.
It is this feature that enables us to reduce from \SAT~to approximating the partition function of antiferromagnetic two-spin systems in the non-uniqueness region.

Restricted to bipartite graphs, it appears that supporting nearly-independent phase-correlated spins alone
is not enough to imply \BIS-hardness.
It was shown that Sly's gadget is applicable to the antiferromagnetic Ising model without an external field by Galanis et al.~\cite{GSV12}.
However, such a system  has
an FPRAS on bipartite graphs.
The reason is that this system is perfectly symmetric on bipartite graphs
and therefore can be translated into a ferromagnetic Ising system,
whose partition function can be approximated using the FPRAS of
Jerrum and Sinclair \cite{JS93} (see Corollary~\ref{cor:Ising:w/o} in Section \ref{sec:sym} for
details).
To get around this perfectly symmetric case,
we introduce the second key concept called unary symmetry breaking.
Unary symmetry breaking is about the availability of non-trivial ``unary weights" which can be built in the model. The availability of unary weights is relevant to counting complexity, for example in the context of counting CSPs. Unary symmetry breaking simply means that some (any) non-trivial unary weight can be built. 
Formal definitions of the two notions -- nearly-independent phase-correlated spins
and unary symmetry breaking -- can be found in Section~\ref{sec:def}.
Our main technical theorem is the following.

\begin{theorem}
  Suppose a tuple of parameters $(\beta,\gamma,\lambda,\Delta)$
    with $\beta \gamma \neq 1$ and $\Delta \geq 3$
  supports nearly-independent phase-correlated spins and unary symmetry breaking.
  Then the partition function
  (\ref{eq:pfdef})
  of two-spin systems $(\beta,\gamma,\lambda)$ is \BIS-hard to approximate on bipartite graphs with
  maximum degree $\Delta$.
    \label{thm:main}
\end{theorem}

Previous hardness proofs
for the problem~\IS\ and for the problem of estimating the
partition function of
antiferromagnetic 2-spin systems typically reduce from {\sc Max-Cut} or
from the problem of
counting certain types of cuts \cite{GJP03,Sly10,SS12}.
However such a technique sheds little light in
the bipartite setting
as cut problems are trivial on bipartite graphs.
Reductions between \BIS-equivalent problems typically involve transformations that ``blow up'' vertices and edges
into sets of vertices that are completely connected, so they do not apply to bounded-degree graphs either.

A key property of Sly's gadget is that either phase occurs with probability bounded below by an inverse polynomial.
This bound is sufficient in Sly's setting to reduce from {\sc Max-Cut}, but it is not enough to construct AP-reductions for our use.
We resolve this issue by introducing a balancing construction.
The construction takes two copies of a gadget with nearly-independent phase-correlated spins,
and produces a new gadget with
similarly-correlated spins,
but in the new gadget the two phases occur with probability close to $1/2$.

The proof of Theorem~\ref{thm:main} utilizes an intermediate problem,
that is, computing the partition function of antiferromagnetic Ising systems with non-uniform external fields on bipartite graphs.
A non-uniform external field means that the instance specifies a subset of vertices on which the external field acts.
A 2-spin system with
a non-uniform external field is very similar to a Boolean \#CSP with
one binary symmetric non-negative valued function (corresponding to edge weights)
and one unary non-negative valued function (corresponding to vertex weights) (see, for example \cite{CLX09a}).

Our reduction implements
an external field,
and this is where unary symmetry breaking comes into play.
As discussed earlier, the partition function of the Ising model without an external field has
an FPRAS, so the unary symmetry breaking gadget seems necessary.
In fact, we show (Lemma~\ref{lem:sym:break}) 
that unary symmetry breaking holds for all 2-spin systems except for the Ising model without
an external field
or degenerate systems (i.e.,
systems satisfying
$\beta\gamma=1$).
We also prove (Lemma \ref{lem:indpt-phases}) that all antiferromagnetic 2-spin systems support
nearly-independent phase-correlated spins in the non-uniqueness region.
Taking Lemmas~\ref{lem:indpt-phases} and~\ref{lem:sym:break} together with Theorem \ref{thm:main} yields our main result:

\begin{theorem}
\label{thm:BIS-nonuniq}
For all tuples of parameters $(\beta,\gamma,\lambda,\Delta)$ with
$\Delta\geq 3$ and $\beta\gamma<1$, except for the case $(\beta=\gamma, \lambda=1)$, if the infinite $\Delta$-regular tree $\TreeD$ is in
the non-uniqueness region then approximating the partition
function   (\ref{eq:pfdef}) on bipartite graphs with maximum degree $\Delta$ is \BIS-equivalent under AP-reductions.
\end{theorem}

There is an FPRAS for the exceptional case ($\beta=\gamma,\lambda=1$) (Corollary~\ref{cor:Ising:w/o}).
Let us now survey the approximability picture that this theorem helps establish.
For general antiferromagnetic 2-spin models with soft constraints (i.e., $\beta\gamma>0$),
non-uniqueness holds if and only if $\sqrt{\beta\gamma}<\frac{\Delta-2}{\Delta}$ and $\lambda\in(\lambda_1,\lambda_2)$
for some critical values $\lambda_1$ and $\lambda_2$ depending on $\beta$, $\gamma$, and $\Delta$ (see \cite[Lemma 5]{LLY13}).
Hence, for all $\beta,\gamma>0$ where $\beta\gamma<1$, and all $\Delta\geq 3$ the following holds:
\begin{enumerate}
\item If $\sqrt{\beta\gamma}>\frac{\Delta-2}{\Delta}$, for all $\lambda$,
there is an FPTAS to approximate the partition function for $\Delta$-regular graphs \cite{SST12,LLY13} (this extends to
graphs of maximum degree $\Delta$ in an appropriate sense, see \cite{LLY13} for details).
\item If $\sqrt{\beta\gamma}<\frac{\Delta-2}{\Delta}$, then there exists $0<\lambda_1<\lambda_2$ so that:
\begin{enumerate}
\item
For all $\lambda\not\in [\lambda_1,\lambda_2]$,
there is an FPTAS to approximate the partition function for $\Delta$-regular graphs \cite{SST12,LLY13} (this again extends in
an appropriate sense to
graphs of maximum degree $\Delta$ \cite{LLY13}).
\item
For all $\lambda\in (\lambda_1,\lambda_2)$,
it is \#SAT-hard to approximate the partition function on $\Delta$-regular graphs \cite{SS12}.
\item
For all $\lambda\in (\lambda_1,\lambda_2)$,
it is \BIS-hard to approximate the partition function on bipartite graphs of maximum degree $\Delta$ (Theorem \ref{thm:BIS-nonuniq} in this paper).
\end{enumerate}
\end{enumerate}

For the particular case of the hard-core model the critical value (i.e., critical activity $\lambda_c(\Delta)$) is more easily stated.
For the hard-core model (i.e., $\beta=0$ and $\gamma=1$) Kelly \cite{Kel85} showed that non-uniqueness holds on $\TreeD$
if and only if $\lambda>\lambda_c(\Delta)\colon=\frac{(\Delta-1)^{\Delta-1}}{(\Delta-2)^\Delta}$.
As a consequence we get the following corollary for the hard-core model.

\begin{corollary}
For all $\Delta\geq 3$, all $\lambda>\lambda_c(\Delta) =\frac{(\Delta-1)^{\Delta-1}}{(\Delta-2)^{\Delta}}$, it is \BIS-hard to approximate the partition function of the hard-core model on bipartite graphs of maximum degree $\Delta$.
  \label{cor:hardcore}
\end{corollary}

We also get a corollary concerning the more
general partition
function as long as $\beta$ and $\gamma$ are less than~$1$
and 
the degree bound~$\Delta$ is sufficiently large.
For and $\beta$ and $\gamma$ satisfying
$0<\beta,\gamma<1$  and any $\lambda>0$, there exists a $\Delta$ such that
$(\beta,\gamma,\lambda)$ is in the non-uniqueness region of $\TreeD$ \cite[Lemma 21.2]{LLY13}.
This implies the following corollary.

\begin{corollary}
  For every $0<\beta,\gamma<1$ and $\lambda>0$,
  there exists a $\Delta$ such that
  the 2-spin system with parameters $\beta,\gamma$
  and with uniform or non-uniform external field $\lambda$ on bipartite graphs with degree bound $\Delta$ is \BIS-equivalent under AP-reductions,
  except when
  $\beta=\gamma$ and $\lambda=1$, in which case it  has an FPRAS.
  \label{cor:general-graphs}
\end{corollary}

More generally, for antiferromagnetic 2-spin systems we get the following
picture for the complexity of approximating the partition function on general graphs.
As usual there is a difficulty classifying the complexity of approximating the partition function at 
the boundary between uniqueness and non-uniqueness.
To address this issue, for parameters $(\beta,\gamma,\lambda,\Delta)$, \cite{LLY13} define a 
notion of {\em up-to-$\Delta$ unique} which is equivalent to the parameters lying 
in the interior of the uniqueness region for the infinite $(d-1)$-ary tree $\Treed$ for all $3\leq d\leq\Delta$ 
(see Definition 7 in~\cite{LLY13}). 
Moreover, the parameters $(\beta,\gamma,\lambda)$ satisfy {\em $\infty$-strict-uniqueness} if 
it is up-to-$\infty$ unique.\footnote{To be precise, the notion of $\infty$-strict-uniqueness is called universally unique in \cite[Definition 7]{LLY13}).}
On the other side, we say the parameters $(\beta,\gamma,\lambda)$ satisfy {\em $\infty$-non-uniqueness} 
if for some $\Delta\geq 3$ the tree $\TreeD$ has non-uniqueness.
The only gap between the notions of $\infty$-strict-uniqueness and $\infty$-non-uniqueness is
the case when the parameters $(\beta,\gamma,\lambda)$ are at the uniqueness/non-uniqueness threshold of $\TreeD$
for some $\Delta$.

The following result detailing the complexity for general graphs is now established.

\begin{corollary}
\label{cor:dichotomy}
For all tuples of parameters $(\beta,\gamma,\lambda)$ with
$\beta\gamma<1$, the following holds:
\begin{enumerate}
\item If the parameters satisfy $\infty$-strict-uniqueness then there is a FPTAS for the partition function for all graphs \cite[Theorem 2]{LLY13}.
\item If the parameters satisfy $\infty$-non-uniqueness then:
\begin{enumerate}
\item it is \#SAT-hard to approximate the partition function on graphs \cite{SS12}.
\item except when $\beta=\gamma$ and $\lambda=1$,
it is \BIS-hard to approximate the partition function on bipartite graphs (Theorem \ref{thm:BIS-nonuniq} in this paper).
\end{enumerate}
\end{enumerate}
\end{corollary}

A recent paper of Liu et al. \cite{LLZ14} 
shows that our Theorem~\ref{thm:main} can also be
used to analyse the complexity of \emph{ferromagnetic} partition functions
(where $\beta \gamma > 1$).
In particular, it uses Theorem \ref{thm:main} to show \BIS-hardness for
approximating the partition function for
ferromagnetic 2-spin systems when $\beta\neq\gamma$  
for sufficiently large external field $\lambda$.
An interesting problem that remains
open is   to prove \BIS-hardness for the entire non-uniqueness region for ferromagnetic
2-spin systems with $\beta\neq\gamma$.

The remainder of the paper is organized as follows.
We begin by formally defining AP-reductions and the problems that we use
in our reductions in Section \ref{sec:prelim}.
In Section \ref{sec:def} we formally define
nearly-independent phase-correlated spins and unary symmetry breaking, and
present the main lemmas concerning these concepts.
We prove Theorem \ref{thm:main} in Section \ref{sec:proof-main}.
Finally, in Sections \ref{sec:proof-balanced-gadget}, \ref{sec:sym}, and \ref{sec:indpt-phases}
we prove the various lemmas about the two
key concepts that are presented in Section \ref{sec:def}.

\section{Approximation-Preserving Reductions and \BIS}
\label{sec:prelim}

We are interested in the complexity of approximate counting.
Let $\Sigma$ be a finite alphabet.
We want to approximate the value of a function $f:\Sigma^*\rightarrow \mathbb{R}$.
A \emph{randomized approximation scheme} is an algorithm that takes an instance $x\in\Sigma^*$
and a rational error tolerance $\varepsilon>0$ as inputs,
and outputs a rational number $z$ such that, for every $x$ and~$\epsilon$,
\[\Pr[e^{-\varepsilon}f(x)\leq z \leq e^{\varepsilon}f(x)]\geq\frac{3}{4}.\]
A
\emph{fully polynomial randomized approximation scheme} (FPRAS) is a randomized approximation scheme
which runs in time bounded by a polynomial in $|x|$ and $\varepsilon^{-1}$.
Note that the quantity $\frac{3}{4}$ can be changed to any value in the interval
$(\frac{1}{2}, 1)$ or even $1 -2^{-n^c}$ for a problem of size $n$
without changing the set of problems
that have fully polynomial randomized approximation schemes
since the higher accuracy can be achieved
with only polynomial delay by taking a majority vote of multiple samples.

Dyer \emph{et al.} \cite{DGGJ03} introduced the notion of approximation-preserving reductions.
Suppose $f$ and $g$ are two functions from $\Sigma^*$ to $\mathbb{R}$.
An \emph{approximation-preserving reduction} (AP-reduction) from $f$ to $g$ is a randomized algorithm
$\mathcal{A}$ to approximate $f$ using an oracle for $g$.
The algorithm $\mathcal{A}$ takes an input $(x,\varepsilon)\in\Sigma^*\times (0,1)$,
and satisfies the following three conditions:
(i) every oracle call made by $\mathcal{A}$ is of the form $(y,\delta)$, where $y\in\Sigma^*$ is an instance of $g$,
and $0<\delta<1$ is an error bound satisfying $\delta^{-1}\leq {\rm poly}(|x|,\varepsilon^{-1})$;
(ii) the algorithm $\mathcal{A}$ meets the specification for being a randomized approximation scheme for $f$
whenever the oracle meets the specification for being a randomized approximation scheme for $g$;
(iii) the run-time of $\mathcal{A}$ is polynomial in $|x|$ and $\varepsilon^{-1}$.

If an AP-reduction from $f$ to $g$ exists, we write $f$\AP$g$,
and say that $f$ is \emph{AP-reducible} to $g$.
If $f$\AP$g$ and $g$\AP$f$,
then we say that $f$ and $g$ are \emph{AP-interreducible} or \emph{AP-equivalent},
and write $f$\APequiv$g$.
The problem \BIS\ is defined as follows.\\

\noindent{\bf Name.} \BIS.\\
\noindent{\bf Instance.} A bipartite graph $B$.\\
\noindent{\bf Output.} The number of independent sets in $B$.\\

In this paper, we are interested in 2-spin systems over bounded degree bipartite graphs parametrized by a tuple
$(\beta,\gamma,\lambda)$.
We say a real number $z$ is \emph{efficiently approximable} if there is an FPRAS for the problem of computing $z$.
Throughout the paper we only deal with non-negative real parameters that are efficiently approximable.
For efficiently approximable non-negative real numbers $\beta,\gamma,\lambda$ and a positive integer $\Delta$,
we define the problem of computing the partition function of the 2-spin system $(\beta,\gamma)$
with external field $\lambda$ on bipartite graphs of bounded degree $\Delta$, as follows.\\

\noindent{\bf Name.} {\sc Bi-(M-)2-Spin}$(\beta,\gamma,\lambda,\Delta)$.\\
\noindent{\bf Instance.} A bipartite (multi)graph $B=(V,E)$ with degree bound $\Delta$.\\
\noindent{\bf Output.} The quantity
\[Z_B(\beta,\gamma,\lambda)=\sum_{\sigma:V\rightarrow\{0,1\}}\lambda^{\sum_{v\in V}\sigma(v)}
\prod_{(v,u)\in E}\beta^{(1-\sigma(v))(1-\sigma(u))}\gamma^{\sigma(v)\sigma(u)}.\]

Notice that we also introduced a multigraph version of the same problem.
It will be useful later.
We drop the parameter $\Delta$ when there is no degree bound,
that is, \BSpin{\beta}{\gamma}{\lambda} is the same as \BBSpin{\beta}{\gamma}{\lambda}{\infty}.

We also found the notion of non-uniform external field useful in the reductions.
The following problems are introduced as intermediate problems.
We also introduce a multigraph version, but as intermediate problems we do not need the bounded degree variant.\\

\noindent{\bf Name.} {\sc Bi-(M-)Nonuniform-2-Spin}$(\beta,\gamma,\lambda)$.\\
\noindent{\bf Instance.} A bipartite (multi)graph $B=(V,E)$ and a subset $U\subseteq V$.\\
\noindent{\bf Output.} The quantity
\[Z_{B,U}(\beta,\gamma,\lambda)=\sum_{\sigma:V\rightarrow\{0,1\}}\lambda^{\sum_{v\in U}\sigma(v)}
\prod_{(v,u)\in E}\beta^{(1-\sigma(v))(1-\sigma(u))}\gamma^{\sigma(v)\sigma(u)}.\]

\section{Key Properties of the Gadget}

\label{sec:def}

In this section we define two key concepts: nearly-independent phase-correlated spins and unary symmetry breaking.

We first describe the basic setup of a certain gadget.
For positive integers $\Delta$, $t$ and $n$ where $n$ is even and is at least $2t$,
let $T^-$ and $T^+$ be disjoint vertex sets of size~$t$
and let $V^-$ be a size-$n/2$ superset of~$T^-$
and $V^+$ be a size-$n/2$ superset of~$T^+$ which is disjoint from~$V^-$.
Let $T=T^-\cup T^+$ and $V(t,n)=V^- \cup V^+ $.
Let $\calG(t,n,\Delta)$ be the set of bipartite graphs with vertex partition $(V^-,V^+)$ in which every vertex has degree at most~$\Delta$
and every vertex in $T$ has degree at most~$\Delta-1$.
We refer to the vertices in $T$ as ``terminals''.
Vertices in $T^+$ are ``positive terminals'' and vertices in~$T^-$ are ``negative terminals''.
 
When the gadget $G$ is drawn from $\calG(t,n,\Delta)$, we use the notation
$T(G)$ to refer to the set of terminals.
Each configuration $ \sigma\colon  V(t,n) \rightarrow \{0,1\}$ is assigned a unique phase $Y(\sigma)\in\phases$.
Roughly in our applications of the
definitions below
the phase of a configuration $\sigma$ is $\pi$ if
$V^\pi$ contains more vertices with spin $1$ than does $V^{-\pi}$.

We define measures $Q^+$ and $Q^-$.
Fix some $0<q^-<q^+<1$. For any positive integer~$t$,
\begin{itemize}
\item $Q^+$ is the distribution on configurations
$\tau\colon T\rightarrow\{0,1\}$ such that, for every $v\in T^+$, $\tau(v)=1$ independently with probability~$\qp$
and, for every $v\in V^-$, $\tau(v)=1$ independently with probability~$\qm$, and
\item $Q^-$ is the distribution on configurations $\tau\colon T\rightarrow\{0,1\}$
such that, for every $v\in T^-$, $\tau(v)=1$ independently with probability~$\qp$
and, for every $v\in T^+$, $\tau(v)=1$ independently with probability~$\qm$.
\end{itemize}

To give a rough sense for the values $q^-$ and $q^+$ they will correspond to the marginal probabilities
of the root of an infinite tree obtained by taking limits of finite trees with appropriate boundary
conditions, see Section \ref{sec:indpt-phases} for more details.

To prove the \BIS-hardness we need a gadget where the spins
of the
terminals are drawn from distributions close to $Q^+$ or $Q^-$
conditioned
on the phase $+$ or $-$.  The following definition formalises this notion. The definition  will be crucial for obtaining our $\BIS$-hardness results for antiferromagnetic models in the non-uniqueness region. While the definition is stated for general 2-spin systems, it is most useful for antiferromagnetic systems; a ferromagnetic 2-spin system is not going to satisfy the conditions of the definition since one expects that for such systems $q^-=q^+$ or, equivalently, there is no notion of a binary ``phase".

\begin{definition}
\label{def:sly}
A tuple of parameters $(\beta,\gamma,\lambda,\Delta)$
 {\bf supports nearly-independent phase-correlated spins}
if there are
efficiently-approximable
values $0 < q^- < q^+ < 1$
such that
the following is true.
There are functions $n(t,\epsilon)$, $m(t,\epsilon)$, and $f(t,\epsilon)$, each of which is bounded from above by a polynomial in~$t$ and~$\epsilon^{-1}$,
and for every~$t$ and~$\epsilon$ there is a distribution on graphs in $\calG(t,n(t,\epsilon),\Delta)$
such that a gadget $G=(V,E)$ with terminals $T$ can be drawn from the distribution within $m(t,\epsilon)$ time,
and the probability that the following inequalities hold is at least~$3/4$:
  \begin{enumerate}
	\item The phases are roughly balanced, i.e.,
      \begin{align}
        \label{eqn:rough:balance}
        \Pr_{G;\beta,\gamma,\lambda}(Y(\sigma)=+) \geq \frac{1}{f(t,\epsilon)} \mbox{ and }
        \Pr_{G;\beta,\gamma,\lambda}(Y(\sigma)=-) \geq \frac{1}{f(t,\epsilon)}.
      \end{align}
    \item For a configuration $\sigma\colon V\rightarrow\{0,1\}$ and any $\tau\colon T\rightarrow\{0,1\}$,
      \begin{align} \label{eqn:near:ind}
            \left|
       \frac{\Pr_{G;\beta,\gamma,\lambda}(\sigma|_{T}=\tau\mid Y(\sigma)=+)}
       {Q^+(\tau)}
       -1
       \right|
       \leq  \epsilon
       \mbox{ and }
          \left|
       \frac{\Pr_{G;\beta,\gamma,\lambda}(\sigma|_{T}=\tau\mid Y(\sigma)=-)}
       {Q^-(\tau)}
       -1
       \right|
       \leq  \epsilon.\end{align}
  \end{enumerate} 
\end{definition}

In fact, given a gadget with the above property, one can construct
a gadget where the phases are (nearly) uniformly distributed
as detailed in the following definition.

\begin{definition}
  \label{def:balsly}
  We say that the tuple of parameters $(\beta,\gamma,\lambda,\Delta)$
  supports {\bf balanced} nearly-independent phase-correlated spins
  if Definition \ref{def:sly} holds with \eqref{eqn:rough:balance} replaced by:
\begin{equation}
        \label{eqn:balance}
        \Pr_{G;\beta,\gamma,\lambda}(Y(\sigma)=+) \geq \frac{1-\epsilon}{2} \mbox{  and  }
        \Pr_{G;\beta,\gamma,\lambda}(Y(\sigma)=-) \geq \frac{1-\epsilon}{2},
\end{equation}
where $\eps$ is quantified as in Definition \ref{def:sly}.
  \end{definition}

  The following lemma shows that for essentially all 2-spin systems,
  Definition \ref{def:sly} implies Definition \ref{def:balsly}.  Hence, in the proof
  of Theorem \ref{thm:main} we assume the existence of a gadget with
  balanced phases.

\begin{lemma}
\label{lem:gettingbal}
If the parameter tuple $(\beta,\gamma,\lambda,\Delta)$ with $\beta\gamma\neq 1$ supports nearly-independent phase-correlated spins,
then it supports balanced nearly-independent phase-correlated spins.
\end{lemma}

The main technical result for proving \BIS-hardness for 2-spin
antiferromagnetic systems in the tree non-uniqueness region
is the following lemma, which is proved in Section \ref{sec:indpt-phases}.

\begin{lemma}
\label{lem:indpt-phases}
For all $\Delta\geq 3$, all $\beta,\gamma,\lambda>0$ where $\beta\gamma < 1$,
if the infinite $\Delta$-regular tree $\TreeD$ is in the non-uniqueness
region then the tuple of parameters $(\beta,\gamma,\lambda,\Delta)$
supports balanced nearly-independent phase-correlated spins.
\end{lemma}

The second property of the gadget is the notion of unary symmetry breaking which is relatively simple.

\begin{definition}
\label{def:symmetry}
We say that a tuple of parameters $(\beta,\gamma,\lambda,\Delta)$
{\bf supports unary symmetry breaking}
if there is a bipartite graph~$H$
whose vertices have degree at most~$\Delta$
which has a distinguished degree-$1$ vertex~$v_H$
such that $\Pr_{H;\beta,\gamma,\lambda}(\sigma_{v_H}=1)
\not\in\{0,\lambda/(1+\lambda),1\}$.
\end{definition}

We will prove in Section~\ref{sec:sym}
that unary symmetry breaking holds for all 2-spin models except for two cases.

\begin{lemma}
  Assume $\Delta\geq3$.
  The parameters $(\beta,\gamma,\lambda,\Delta)$ support unary symmetry breaking
  unless (i)~$\beta\gamma=1$ or (ii)~ $\beta=\gamma$ and $\lambda=1$.
  \label{lem:sym:break}
\end{lemma}

\section{General Reduction}
\label{sec:proof-main}

In this section we prove Theorem \ref{thm:main}.
We first show how the two notions of ``nearly-independent phase-correlated spins'' and ``unary symmetry breaking'' lead to \BIS-hardness.

\subsection{An Intermediate Problem}

The goal of this section is to show that it is \BIS-hard to approximate the partition function of
antiferromagnetic Ising models with non-uniform non-trivial external fields on bipartite graphs.

\begin{lemma}
  For any $0<\alpha<1$, $\lambda>0$ and $\lambda\neq 1$,
  \center{\BIS ~\AP ~\BMNSpin{\alpha}{\alpha}{\lambda}.}
  \label{lem:BIS:Ising}
\end{lemma}

\begin{proof}
By flipping $0$ to $1$ and $1$ to $0$ for each configuration $\sigma$,
we see that \BMNSpin{\alpha}{\alpha}{\lambda}~is in fact the same as \BMNSpin{\alpha}{\alpha}{1/\lambda}.
Hence we may assume $\lambda<1$.

Let $M = \left(\begin{matrix}
 \alpha &  1\\
 1 &  \alpha
\end{matrix}\right),$
and
$\left(\begin{matrix}
\rho_0\\
\rho_1
\end{matrix}\right) = M
\left(\begin{matrix}
 1\\
 \lambda
\end{matrix}\right) =
\left(\begin{matrix}
 \alpha+\lambda\\
 1+\alpha \lambda
\end{matrix}\right).$
Note that $\alpha<1$ and
$\lambda < 1$, so
$\rho_1 > \rho_0$.

Let $B=(V,E)$ be an input to~$\BIS$ with $n=|V|$ and $m=|E|$.
Let $I_B$ be the number of independent sets of~$B$.
Let $\epsilon$ be the desired accuracy of the reduction.
We will construct an instance~$B'=(V',E')$ with a specified vertex subset $U\subset V'$ for
\BMNSpin{\alpha}{\alpha}{\lambda}~such that
\[\exp\left( -\frac{\epsilon}{2} \right) I_B \leq \frac{Z_{B',U}(\alpha,\alpha,\lambda)}{C} \leq \exp\left( \frac{\epsilon}{2} \right) I_B,\]
where $C$ is a quantity that is easy to approximate.
Therefore it suffices to call oracle \BMNSpin{\alpha}{\alpha}{\lambda}~on $B'$ with the specified subset $U$ with accuracy $\frac{\epsilon}{4}$
and approximate $C$ within $\frac{\epsilon}{4}$.

The construction of $B'$ involves two positive integers $t_1$ and $t_2$.
Let $t_1$ be the least positive integer
such that
\begin{align}
  \label{eqn:t1}
  \alpha^{2t_1} \leq \frac{\epsilon}{6\cdot 2^{n}}.
\end{align}
Note that $t_1$ depends on $n$ and $\epsilon$ and
there is a polynomial $p$ in~$n$ and $\epsilon^{-1}$
such that $t_1 \leq p(n,\epsilon^{-1})$.
Let $t_2$ be the least positive integer depending on~$n$, $\epsilon$ and~$t_1$
such that
\begin{align}
  \label{eqn:t2}
    \left(\frac{\rho_0}{\rho_1}\right)^{t_2}
	\leq \frac{\alpha^{t_1 m} \cdot {\epsilon} }{6\cdot 2^{2 t_1 m + n}}.
\end{align}
Once again, $t_2$ is bounded from above by a polynomial in~$n$ and~$\epsilon^{-1}$.

Given the integers~$t_1$ and~$t_2$, the graph~$B'$ is constructed as follows (see Figure~\ref{fig:toy} for an illustrative diagram).
Let $W_v = \{w_{v,j} \mid 1\leq j\leq t_1 \deg(v)\}$ for each $v\in V$ where $\deg(v)$ is the degree of $v$ in $B$.
Let $U_{v,j}=\{u_{v,j,k} \mid 1\leq k\leq t_2\}$ for any $v\in V$ and $1\leq j\leq t_1\deg(v)$.
Let \[W=\bigcup_{v\in V} W_v \mbox{\qquad and \qquad} U=\bigcup_{v\in V} \bigcup_{1\leq j\leq t_1\deg(v)} U_{v,j}.\]
The vertex set of $B'$ is $V'=V\cup U\cup W$.
Note that $|W|=2t_1 m$ and $|U|=2t_1t_2m$.

We add $t_1$ parallel edges in $B'$ between $u$ and $v$ for each $(u,v)\in E$ and
add edges between $v$ and every vertex in $W_v$, and between $w_{v,j}$ and every vertex in $U_{v,j}$ for each $v\in V$ and $1\leq j\leq t_1\deg(v)$.
Formally the edge set of $B'$ is
\[E'=\left(\biguplus_{1\leq i\leq t_1}E\right)\cup
\bigcup_{v\in V}E_v\cup
\bigcup_{\substack{v\in V\\ 1\leq j\leq t_1\deg(v)}}E_{v,j},\]
where $\biguplus$ denotes a disjoint union as a multiset of $t_1$ copies of $E$,
$E_v=\{(v,w)|w\in W_v\}$ and $E_{v,j}=\{(w_{v,j},u)|u\in U_{v,j}\}$ for each $v$ and $j$.

\begin{figure}[h]
\includegraphics[trim= 20mm 195mm 20mm 20mm, clip]{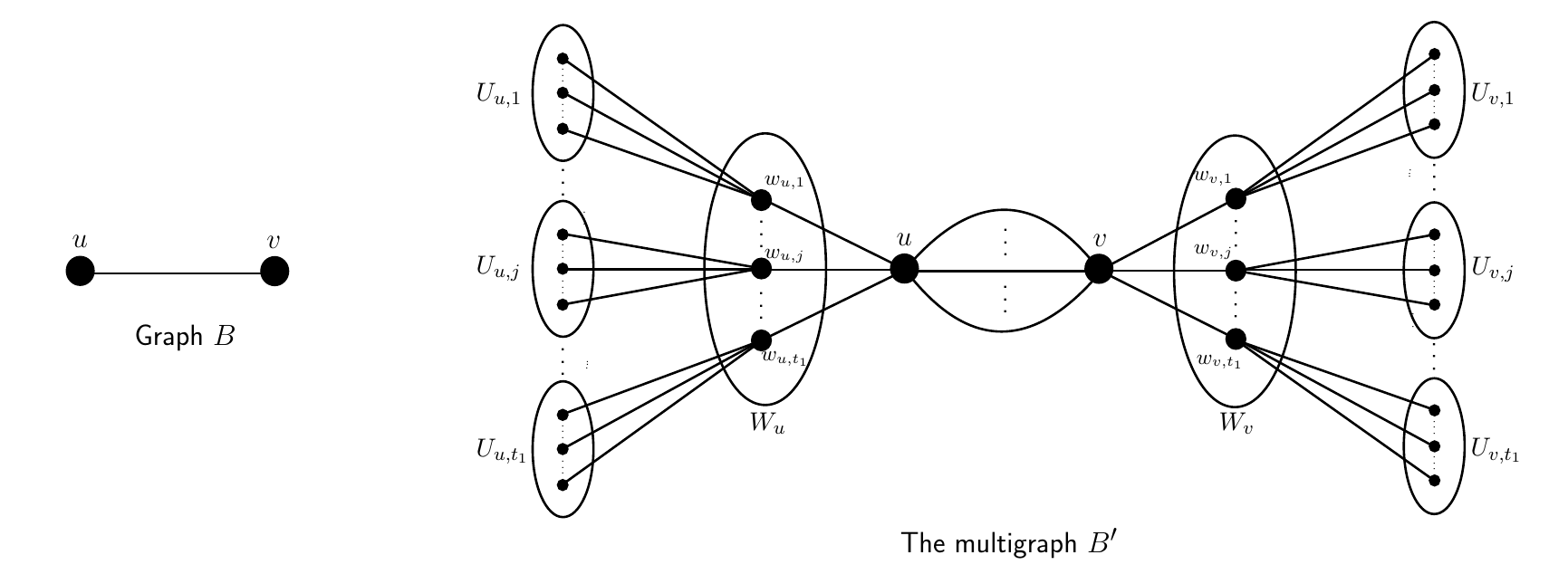}
\caption{An illustrative diagram of the reduction in the proof of Lemma~\ref{lem:BIS:Ising}. As an example, we consider the (bipartite) graph $B$ in the left part of the picture to be the input to $\BIS$. The respective input $B'$ to \BMNSpin{\alpha}{\alpha}{\lambda}  is shown in the right part of the picture. For this graph $B$, where $\mathrm{deg}(u)=\mathrm{deg}(v)=1$,  $t_1$ is the number of parallel edges between $u$ and $v$ in $B'$, as well as the cardinalities of the sets $W_u$ and $W_v$ (where $t_1$ is an appropriate integer, cf. \eqref{eqn:t1}). Also, $t_2$ (cf. \eqref{eqn:t2}) equals the cardinalities of the sets $U_{u,j}$ and $U_{v,j}$ for $j=1,\hdots,t_1$. For a general graph $B$ and a vertex $u$ in $B$, the cardinality of the set $W_u$ depends on the degree of the vertex $u$ and equals $t_1\mathrm{deg}(u)$. Other than that, the construction of the multigraph $B'$ is analogous by iterating the construction in the figure for each edge in $B$.}\label{fig:toy}
\end{figure}

Let
$C = \rho_1^{2 t_1 t_2 m }  {\alpha}^{t_1 m}$
and $N=\left(\begin{matrix}
 1 &  1\\
 1 &  \alpha^{2 t_1}
\end{matrix}\right).$

For each $\sigma: V\cup W \rightarrow \{0,1\}$,
let $w(\sigma)$ be the contribution to $Z_{B',U}(\alpha,\alpha,\lambda)$ of configurations that are consistent with $\sigma$.
First consider configurations $\sigma$ such that $\sigma(w)=1$ for all $w\in W$.
Denote by $\Sigma$ the set of all such configurations on $V\cup W$.
Then for $\sigma\in\Sigma$, we have
\begin{align}
  w(\sigma) &= \rho_1^{t_2 |W|}
  \prod_{(u,v)\in E} ( M_{1,\sigma(u)} M_{\sigma(u),\sigma(v)} M_{\sigma(v),1})^{t_1} \label{eq:thnnht}\\
  &= C \prod_{(u,v)\in E}
  N_{\sigma(u),\sigma(v)}.\notag
\end{align}
To see \eqref{eq:thnnht}, note that $W$ separates the vertices in $U$ from the rest of the vertices of $B'$. Since $\sigma$ specifies the spins of all vertices in $W$, the weight $w(\sigma)$ is the product of: (i) the  weights induced on edges that do not involve vertices in $U$, and (ii) the weights induced on edges with endpoints in $U$. Item (i) is just $\prod_{(u,v)\in E} ( M_{1,\sigma(u)} M_{\sigma(u),\sigma(v)} M_{\sigma(v),1})^{t_1}$. For Item (ii),  we need to sum over all possible assignments on $U$; the factor $\rho_1^{t_2 |W|}$ gives the value of this sum as can be seen by expanding the expression \[\prod_{(w,u)\in E'\cap (W\times U)}(1+\alpha \lambda)=(1+\alpha \lambda)^{t_2 |W|}=\rho_1^{t_2 |W|}.\] 

Let $\Sigma^{ind}\subset\Sigma$ be the subset of configurations which induce an independent set on the vertices $V$
and $Z^{ind}$ be its contribution to $Z_{B',U}(\alpha,\alpha,\lambda)$.
Let $\Sigma^{bad}=\Sigma\backslash\Sigma^{ind}$ and $Z^{bad}$ be its contribution.
If~$\sigma\in\Sigma^{ind}$ then $w(\sigma) = C$.
Otherwise, $w(\sigma) \leq \alpha^{2 t_1} C$.
It implies
\begin{align}
  Z^{ind}=I_B\cdot C \mbox{\qquad and\qquad}Z^{bad}\leq 2^n\alpha^{2 t_1} C\leq \frac{\epsilon}{6}\cdot C,
  \label{eqn:sigma}
\end{align}
since $t_1$ satisfies Eq.~(\ref{eqn:t1}).

Next consider configurations~$\sigma$ on $V\cup W$
such that $\sigma(w)=0$ for at least one $w\in W$.
Denote this set by $\Sigma'$ and its contribution by $Z^{small}$.
Then for $\sigma\in\Sigma'$,
\[w(\sigma) \leq \left(\rho_0 \rho_1^{|W|-1}\right)^{t_2}
  \leq \left(\frac{\rho_0}{\rho_1}\right)^{t_2} \rho_1^{t_2 |W|}
= {\left(\frac{\rho_0}{\rho_1}\right)}^{t_2} \frac{C} {\alpha^{t_1m}}.\]
It implies
\begin{align}
  Z^{small}\leq 2^{2 t_1 m + n} {\left(\frac{\rho_0}{\rho_1}\right)}^{t_2} \frac{C} { \alpha^{t_1m}}\leq \frac{\epsilon}{6}\cdot C,
  \label{eqn:sigmaprime}
\end{align}
since $|\Sigma'|\leq 2^{2 t_1 m + n}$ and $t_2$ satisfies Eq.~(\ref{eqn:t2}).

By Eq.~(\ref{eqn:sigma}) and Eq.~(\ref{eqn:sigmaprime}) we have
\begin{align*}
  Z_{B',U}(\alpha,\alpha,\lambda)
  & = Z^{ind}+Z^{bad}+Z^{small} \\
  &\leq I_B \cdot C + \frac{\epsilon}{6}\cdot C + \frac{\epsilon}{6}\cdot C  \\
  &\leq \exp\left(\frac{\epsilon}{3}\right)I_B \cdot C,
\end{align*}
and clearly $Z_{B',U}(\alpha,\alpha,\lambda) \geq I_B\cdot C$.
It is also clear that $C$ can be approximated accurately enough given FPRASes for $\lambda$ and $\alpha$.
This finishes our proof.
\end{proof}

\subsection{Simulating the Antiferromagnetic Ising Model}

In this section we prove the following lemma.

\begin{lemma} 
Suppose $\beta$, $\gamma$ and $\lambda$ are efficiently
approximable reals satisfying
$\beta,\gamma\geq 0$, $\lambda>0$
and $\beta \gamma \neq 1$.
Suppose that
$\Delta$ is either an integer that is at least~$3$ or
 $\Delta=\infty$ (indicating that we do not have a degree bound).
If $(\beta,\gamma,\lambda,\Delta)$ supports nearly-independent phase-correlated spins and unary symmetry breaking,
  then there exist efficiently approximable $0<\alpha<1$ and $\lambda'>0$
  such that $\lambda'\neq 1$ and
  \center   {\BMNSpin{\alpha}{\alpha}{\lambda'} ~\AP ~\BBSpin{\beta}{\gamma}{\lambda}{\Delta}.}
  \label{lem:main}
\end{lemma}

\begin{proof}
  We prove the antiferromagnetic case first, that is, $\beta\gamma<1$.
  $\alpha$ and $\lambda'$ are chosen as follows.
  Recall that $M=\left(\begin{matrix}
  \beta & 1\\
  1 & \gamma
  \end{matrix}\right)$ and $M^+=\left(\begin{matrix}
  1-q^- & q^-\\
  1-q^+ & q^+
  \end{matrix}\right)$.
  Let $N=M^+ M (M^+)^{\texttt T}=\left(\begin{matrix}
          N_{--} & N_{-+}\\
          N_{+-} & N_{++}
        \end{matrix}\right)$.
  Then $\det(N)=(\beta\gamma-1)(\qp-\qm)^2<0$.
  Therefore $N_{--}N_{++}<N_{-+}N_{+-}$ and let $\alpha=\frac{N_{--}N_{++}}{N_{-+}N_{+-}}<1$.
  Moreover, suppose $H$ is the unary symmetry breaking gadget with distinguished vertex $v_H$.
  Let $\rho=\left(\begin{matrix}
          \rho_0 \\
          \rho_1
        \end{matrix}\right)$
  where $\rho_i$ denote $\Pr_{H;\beta,\gamma,\lambda}(\sigma_{v_H}=i)$ for spin $i\in\{0,1\}$
  and $\rho_0+\rho_1=1$.
  Let $\rho'=\left(\begin{matrix}
          \rho'_0 \\
          \rho'_1
        \end{matrix}\right)
		=M^+\left(\begin{matrix}
          \rho_0 \\
		  \rho_1/\lambda
        \end{matrix}\right)$,
  and $\lambda'=\frac{\rho'_1}{\rho'_0}$.
  It is easy to verify that $\lambda'\neq 1$ as $\rho_1\neq\lambda/(1+\lambda)$ by the unary symmetry breaking assumption.

  Given $0<\epsilon<1$ and a bipartite multigraph $B=(V,E)$ with a subset $U\subseteq V$ where $|V|=n$, $|E|=m$, and $|U|=n'$,
  our reduction first constructs a bipartite graph $B'$ with degree at most $\Delta$.
  The construction of $B'$ involves a gadget $G$.
  Since $(\beta,\gamma,\lambda,\Delta)$ supports nearly-independent phase-correlated spins,
  by Lemma~\ref{lem:gettingbal} $(\beta,\gamma,\lambda,\Delta)$ also supports balanced nearly-independent phase-correlated spins.
  Therefore we draw $G\sim\calG(t,n(t,\epsilon'),\Delta)$ such that Eq.~(\ref{eqn:balance}) and Eq.~(\ref{eqn:near:ind}) hold with probability at least $3/4$,
  where $t=m+1$ and $\epsilon'=\frac{\epsilon}{8n}$.
  Assume $G$ satisfies them and otherwise the reduction fails.
  We will construct $B'$ such that
  \begin{align*}
	\exp\left( -\frac{\epsilon}{2} \right) Z_{B,U}(\alpha,\alpha,\lambda')
	\leq \frac{Z_{B'}}{\left(N_{+-}N_{-+}\right)^{m}
    \left(\rho_0' Z_H\right)^{n'}\left(\frac{Z_G}{2}\right)^n}
	\leq \exp\left( \frac{\epsilon}{2} \right) Z_{B,U}(\alpha,\alpha,\lambda') ,
  \end{align*}
  where we use the abbreviated expressions $Z_{B'}=Z_{B'}(\beta,\gamma,\lambda)$, $Z_H=Z_H(\beta,\gamma,\lambda)$, and $Z_G=Z_G(\beta,\gamma,\lambda)$.
  The lemma follows by one oracle call for $Z_{B'}$ with accuracy $\frac{\epsilon}{6}$,
  one oracle call for $Z_G$ with accuracy $\frac{\epsilon}{6n}$,
  and an approximation of other terms in the denominator with accuracy $\frac{\epsilon}{6}$ using FPRASes for $\qm$, $\qp$, $\beta$,
    $\gamma$ and~$\lambda$.

  The graph $B'$ is constructed as follows.
  For each vertex $v\in V$ we introduce a copy of $G$, denoted by $G_v$
  with vertex set $V(G_v)$.
  Moreover, for each vertex $u\in U$ we introduce a copy of $H$, denoted by $H_u$.
  Whenever a terminal vertex is used in the construction once, we say it is occupied.
  For each $(u,v)\in E$, we connect one currently unoccupied positive (and respectively negative) terminal of $G_u$
  to one currently unoccupied positive (and respectively negative) terminal of $G_v$.
  Denote by $E'$ all these edges between terminals.
  For each $u\in U$, we identify an unoccupied positive terminal of $G_u$ with the distinguished vertex $v_{H_u}$ of $H_u$.
  We denote this terminal by $t_u$.
  The resulting graph is $B'$.
  It is clear that $B'$ is bipartite and has bounded degree $\Delta$.

  Let $ \pconfig\colon V\rightarrow\phases$ be a configuration of the phases of
  the
  $G_v$'s.
    Let $Z_{B'}( \pconfig)$ be the contribution to $Z_{B'}$ from the configurations~$\sigma$
    that are consistent with $\pconfig$ in the sense that, for each $v\in V$,
    $Y(\sigma_{V(G_v)})=\pconfig(v)$.
  Then $Z_{B'}=\sum_{ \pconfig}Z_{B'}( \pconfig)$.
  Let $T$ be the set of all terminals $T=\cup_{v\in V}T(G_v)$
  and $\tau\colon T\rightarrow\{0,1\}$ be a configuration on $T$.
  Let $\tau_{T(G_v)}$ be the configuration $\tau$ restricted to $T(G_v)$.
 Recall that
 for $\pi\in\{-,+\}$,
 $Z_{G_v}^{\pi}(\tau_{T(G_v)})$ is the contribution to $Z_{G_v}$ from
 configurations that  have phase $\pi$ and are consistent with $\tau_{T(G_v)}$. Also,
   \begin{align*}
	\Pr_{G_v;\beta,\gamma,\lambda}(\tau_{T(G_v)}\mid
Y(\sigma_{V(G_v)})=\pi)
	=\frac{Z_{G_v}^\pi(\tau_{T(G_v)})}{Z_{G_v}^\pi}.
  \end{align*}
  Moreover, for each $u\in U$
  and each spin $i\in\{0,1\}$,
  let $Z_{H_u}(i)$ be the contribution to $Z_{H_u}$
  from configurations~$\sigma$
  with $\sigma(t_u)=i$.
Hence,
  \[ \rho_i=\Pr_{H_u;\beta,\gamma,\lambda}( 
  \sigma(t_u)=i)=\frac{Z_{H_u}(i)}{Z_{H_u}}.\]
  We express $Z_{B'}( \pconfig)$ as
  \begin{align*}
	Z_{B'}( \pconfig)=\sum_{\tau\colon T\rightarrow\{0,1\}}w_{E'}(\tau)
	\prod_{v\in V}Z_{G_v}^{\pconfig(v)}(\tau_{T(G_v)})
	\prod_{u\in U}\frac{Z_{H_u}(\tau(t_u))}{\lambda^{\tau(t_u)}},
  \end{align*}
  where $w_{E'}(\tau)$ is the contribution of edges in $E'$ given configuration $\tau$.
  Notice that we divide the last factor by $\lambda$ when $\tau(t_u)=1$ because we counted the vertex weight twice in that case.
  Define $\widetilde{Z_{B'}}(\pconfig)$ to be an approximation version of the partition function
  where on each $T(G_v)$ the
  spins
  are chosen exactly according to
  $Q^{\pconfig(v)}$.
  That is,
  \begin{align}
	\widetilde{Z_{B'}}(\pconfig)&=\sum_{\tau\colon T\rightarrow\{0,1\}}w_{E'}(\tau)
	\prod_{v\in V}Z_{G_v}^{\pconfig(v)}Q^{\pconfig(v)}(\tau_{T(G_v)})
	\prod_{u\in U}\frac{Z_{H_u}(\tau(t_u))}{\lambda^{\tau(t_u)}} \notag \\
	&= \left(\prod_{v\in V}Z_{G_v}^{\pconfig(v)}\right)\cdot
	\left( \sum_{\tau\colon T\rightarrow\{0,1\}}w_{E'}(\tau)
	\prod_{v\in V}Q^{\pconfig(v)}(\tau_{T(G_v)})
	\prod_{u\in U}\frac{Z_{H_u}(\tau(t_u))}{\lambda^{\tau(t_u)}} \right). \label{eqn:tildeZ}
  \end{align}
  Let $\widetilde{Z_{B'}}=\sum_{\pconfig}\widetilde{Z_{B'}}(\pconfig)$.
  Eq.~(\ref{eqn:near:ind}) implies that $Z_{B'}(\pconfig)$ and $\widetilde{Z_{B'}}(\pconfig)$ are close, that is,
  \begin{align}
	(1-\epsilon')^n
	\leq\frac{Z_{B'}(\pconfig)}{\widetilde{Z_{B'}}(\pconfig)}
	\leq(1+\epsilon')^n.
	\label{eqn:ratio:Z:tildeZ}
  \end{align}
  Moreover Eq.~(\ref{eqn:balance}) implies that
  \begin{align}
	\left(\frac{1-\epsilon'}{2}\right)^n\leq
	\frac{\prod_{v\in V}Z_{G_v}^{\pconfig(v)}}{\left( Z_{G} \right)^n}
	\leq\left(\frac{1+\epsilon'}{2}\right)^n.
	\label{eqn:ratio:balance}
  \end{align}
  Notice that here $Z_{G_v}$ is the same for any $v\in V$ as
  the
  $G_v$'s are identical copies of $G$.

  Then we calculate the following quantity given $\pconfig$
  \[ \sum_{\tau\colon T\rightarrow\{0,1\}}w_{E'}(\tau)
	\prod_{v\in V}Q^{\pconfig(v)}(\tau_{T(G_v)})
  \prod_{u\in U}\frac{Z_{H_u}(\tau(t_u))}{\lambda^{\tau(t_u)}}.\]
  As the measure $Q^{\pconfig(v)}$ is i.i.d., we may count the weight of each edge in $E'$ independently.
  Notice that $N_{\phase_1\phase_2}$ is the edge contribution when one end point is chosen with probability $q^{\phase_1}$ and the other $q^{\phase_2}$.
  For an edge $(u,v)\in V$, if $u$ and $v$ are assigned the same phase $+$,
  then an edge in $E'$ connecting one $+$ terminal of $G_u$ and one $+$ terminal of $G_v$ gives a weight of $N_{++}$
  and an edge connecting two $-$ terminals gives $N_{--}$.
  The total weight is $\mu_1=N_{++}N_{--}$.
  Similarly if $u$ and $v$ are assigned the same phase $-$,
  the total weight is $\mu_1$ as well.
  On the other hand if $u$ and $v$ are assigned distinct phases $+$ and $-$,
  the total weight is $\mu_2=N_{+-}N_{-+}$.
  Recall that $\alpha=\frac{\mu_1}{\mu_2}$.
  Moreover, for each $u\in U$, if $\pconfig(u)=+$, then the contribution of $H_u$ is $\rho_1'Z_{H_u}$ and otherwise $\rho_0'Z_{H_u}$.
  Notice that here $Z_{H_u}$ is the same for any $u\in U$ as
  the
  $H_u$'s are identical copies of $H$.
  Recall that $\lambda'=\frac{\rho_1'}{\rho_0'}$.

  Plug these calculation into Eq.~(\ref{eqn:tildeZ}), we have
  \begin{align}
	\widetilde{Z_{B'}}(\pconfig)
	&= \left(\prod_{v\in V}Z_{G_v}^{\pconfig(v)}\right)\cdot	
	\left(
	\mu_1^{m_+(\pconfig)}
	\mu_2^{m-m_+(\pconfig)}
	\left(\rho_1'Z_{H}\right)^{n_+(\pconfig)}
	\left(\rho_0'Z_{H}\right)^{n'-n_+(\pconfig)}
	\right) \notag\\
	&= \mu_2^{m}
    \left(\rho_0'Z_{H}\right)^{n'}
	\left(\prod_{v\in V}Z_{G_v}^{\pconfig(v)}\right)\cdot
	\left(
	\alpha^{m_+(\pconfig)}
	\left(\lambda'\right)^{n_+(\pconfig)}
	\right),
	\label{eqn:tildeZsig}
  \end{align}
  where $m_+(\pconfig)$ denotes the number of edges of which the two endpoints are of the same
  phase
  given $\pconfig$,
  and $n_+(\pconfig)$ denotes the number of vertices in $U$ that are assigned $+$ given $\pconfig$.
  Apply Eq.(\ref{eqn:ratio:balance}) to Eq.(\ref{eqn:tildeZsig}),
  \begin{align}
    (1-\epsilon')^n\left(
	\alpha^{m_+(\pconfig)}
	\left(\lambda'\right)^{n_+(\pconfig)}
	\right)\leq
	\frac{\widetilde{Z_{B'}}(\pconfig)}{
	\mu_2^{m}
    \left(\rho_0' Z_{H}\right)^{n'}
	\left(\frac{Z_{G}}{2}\right)^n}
    \leq (1+\epsilon')^n\left(
	\alpha^{m_+(\pconfig)}
	\left(\lambda'\right)^{n_+(\pconfig)}
	\right).
	\label{eqn:tildeZ:sig}
  \end{align}
  Then we sum over $\pconfig$ in Eq.~(\ref{eqn:tildeZ:sig}),
  \begin{align}
    (1-\epsilon')^n\left(\sum_{\pconfig}
	\alpha^{m_+(\pconfig)}
	\left(\lambda'\right)^{n_+(\pconfig)}\right)
	\leq
	\frac{\widetilde{Z_{B'}}}
	{\mu_2^{m}
    \left(\rho_0' Z_{H}\right)^{n'}
	\left(\frac{Z_{G}}{2}\right)^n}
	&\leq(1+\epsilon')^n\left(\sum_{\pconfig}
	\alpha^{m_+(\pconfig)}
	\left(\lambda'\right)^{n_+(\pconfig)}\right)
    \label{eqn:tildeZB}
  \end{align}
  However notice that $Z_{B,U}(\alpha,\alpha,\lambda')=\sum_{\pconfig}
	\alpha^{m_+(\pconfig)}
	\left(\lambda'\right)^{n_+(\pconfig)}$ by just mapping $-$ to $0$ and $+$ to $1$ in each configuration $\pconfig$.
  Combine Eq.(\ref{eqn:ratio:Z:tildeZ}), and Eq.(\ref{eqn:tildeZB}),
  \begin{align*}
	(1-\epsilon')^{2n} Z_{B,U}(\alpha,\alpha,\lambda')
	\leq \frac{Z_{B'}}{\mu_2^{m}
    \left(\rho_0' Z_{H}\right)^{n'}
	\left(\frac{Z_{G}}{2}\right)^n}
	\leq (1+\epsilon')^{2n}
	Z_{B,U}(\alpha,\alpha,\lambda').	
  \end{align*}
  Recall that $\epsilon'=\frac{\epsilon}{8n}$ and we get the desired bounds.

  The other case is ferromagnetic, that is, $\beta\gamma>1$.
  Notice that in this case $\det(N)=(\beta\gamma-1)(\qp-\qm)^2>0$,
  So we choose $\alpha=\frac{N_{+-}N_{-+}}{N_{++}N_{--}}<1$ and $\lambda'$ to be the same as the antiferromagnetic case.
  The construction of $B'$ is similar to the previous case, with the following change.
  For each $(u,v)\in E$, we connect one unoccupied positive terminal of $G_u$ to one unoccupied negative terminal of $G_v$,
  and vice versa.
  The rest of the construction is the same.
  With this change, given a configuration $\pconfig\colon V\rightarrow\phases$,
  if two endpoints are assigned the same spin, the contribution is $N_{+-}N_{-+}$
  and otherwise $N_{++}N_{--}$.
  Therefore the effective edge weight is $\alpha<1$ when the spins are the same, after normalizing the weight to $1$ when the spins are distinct.
  The rest of the proof is the same.
\end{proof}

\subsection{Completing the Proof of Theorem \ref{thm:main}}
\label{sec:BIS-easiness}

We can now conclude the proof of Theorem \ref{thm:main}.

\begin{proof}[Proof of Theorem \ref{thm:main}]
\BIS-hardness in Theorem \ref{thm:main} follows directly from
Lemma~\ref{lem:BIS:Ising} and Lemma~\ref{lem:main}.
The other direction, \BIS-easiness, follows fairly directly from Theorem 47 of
\cite{CDGJLMR12} (the full version of \cite{CDGJLMR13}).  
An edge in the instance graph can be viewed as a constraint of arity~2.
If $\beta\gamma> 1$, then the constraint on the edge is ``weakly log-supermodular''
and the vertex weight can be viewed as a unary constraint, 
which is taken as given in a ``conservative'' CSP\null.
If $\beta\gamma\leq 1$, then reverse the interpretation of $0$ and $1$
on one side of the bipartition of the instance graph, so that the effective
interaction along an edge is given by the matrix $\big(\begin{smallmatrix}1&\beta\\\gamma&1\end{smallmatrix}\big)$. 
This constraint is also ``weakly log-supermodular'' since $1\cdot1\geq \beta\gamma$.
After the reversing there are two vertex weights $\lambda$ and $\lambda^{-1}$,
which are also allowed for ``conservative'' CSP instances.
\end{proof}

\section{Balanced Nearly-Independent Phase-Correlated Spins}
\label{sec:proof-balanced-gadget}

In this section we prove Lemma \ref{lem:gettingbal} that a gadget with
nearly-independent phases can be used to construct a gadget with balanced phases.

Before proving the lemma, we introduce some
notation.
Let \[M = \left(\begin{matrix}
\beta & 1\\
1 & \gamma
\end{matrix}\right).\]
Let $q^+$ and $q^-$ be the quantities from Definition~\ref{def:sly}
and let
\[M^+ = \left(\begin{matrix}
1-q^- & q^-\\
1-q^+ & q^+
\end{matrix}\right)\]
The two columns of~$M^+$ correspond to spin~$0$ and spin~$1$.
The first row corresponds to the distribution
induced on a positive terminal from~$Q_t^-$ and the
second to the distribution induced from~$Q_t^+$.
Similarly the first row also corresponds to the distribution induced on a negative terminal from $Q_t^+$
and the second to the distribution induced from $Q_t^-$.
Notice that $\det(M^+) = \qp - \qm >0$.

When the parameters~$\beta$, $\gamma$ and $\lambda$ are clear,
we sometimes make the notation more concise, referring
to the partition function as $Z_G$ rather than
as $Z_G(\beta,\gamma,\lambda)$.
Also, given a configuration $\sigma\colon V(G) \to \{0,1\}$
and a subset $S$ of~$V(G)$, we often use the
notation $\sigma_S$ to denote the restriction $\sigma|_S$.
For a gadget
$G$ drawn from
$\calG(t,n(t,\epsilon),\Delta)$,
let $Z_{G}^\phase$ be the contribution of phase $\phase\in\phases$ to the partition function $Z_G$.
Moreover, for a subset
$S\subseteq T(G)$, suppose $\tau_S\colon T(G)\rightarrow\{0,1\}$ is a configuration on terminals in $S$.
Let $Z_G^\phase(\tau_S)$ be the contribution of configurations that are consistent with $\tau_S$ and 
belong
to phase $\phase$,
that is,
\[Z_G^\phase(\tau_S)=\sum_{\substack{\sigma\colon Y(\sigma)=\phase \\
\sigma_{S}=\tau_S}} w(\sigma),\]
where $w(\sigma)$ is the weight of configuration $\sigma$
 defined in (\ref{ourdefw}).
It is easy to
see that for $\phase\in\phases$,
\[\Pr_{G;\beta,\gamma,\lambda}(Y(\sigma)=\phase)=\frac{Z_{G}^\phase}{Z_G},\]
and
\[\Pr_{G;\beta,\gamma,\lambda}(\sigma_{T(G)}=\tau_{T(G)}\mid Y(\sigma)=\phase)=\frac{Z_{G}^\phase(\tau_{T(G)})}{Z_G^\phase}.\]

We are now prepared to prove the lemma which is the focus of this section.

\begin{proof}[Proof of Lemma \ref{lem:gettingbal}]
  Let $\epsilon$ satisfy $0<\epsilon<1$.
  By assumption we may draw a gadget $G$ from $\calG(t+t',n(t+t',\epsilon'),\Delta)$ such that it satisfies Eq.~(\ref{eqn:rough:balance}) and Eq.~(\ref{eqn:near:ind})
  with probability at least 3/4, where $t'$ and $\epsilon'$ will be specified later.
  Assume $G$ does. Otherwise the construction fails.

  We consider first the antiferromagnetic case $\beta\gamma<1$.
  We construct a gadget $K$ such that $K$ satisfies Eq.~(\ref{eqn:balance}) and Eq.~(\ref{eqn:near:ind}).
  We make two copies of $G$, say $G_1$ and $G_2$.
  Let
  the
  terminals of $G_i$ be $T(G_i)=T^+(G_i)\cup T^-(G_i)$ for $i=1,2$.
  For each $\phase\in\phases$, we add a set of edges that form a perfect matching between $t'$ terminals in $T^\phase(G_1)$ and $t'$ terminals in $T^\phase(G_2)$.
  Denote by $P$
  the
  edges of the two perfect matchings.$K$ is the resulting graph.
  Denote by $C_i$ the vertices of $G_i$ that are endpoints of $P$.
  The terminals of $K$ are those $2t$ terminal nodes in $T(G_1)$ that are still unmatched,
  that is $T^\phase(K)=T^\phase(G_1)\backslash C_1$ for $\phase\in\phases$,
  and $T(K)=T^+(K)\cup T^-(K)$.
  Denote by $I$ the terminals of $G_2$ that are unmatched.

  We define the phase of $K$ to be the phase of $G_1$, that is, $K$ is said to have phase $+$ or $-$
  if and only if $G_1$ has the same $+$ or $-$ phase regardless of the phase of $G_2$.
  Let $(\phase_1,\phase_2)$ be a vector denoting the phases of $G_1$ and $G_2$ where $\phase_1,\phase_2\in\phases$.
  Then the $+$ phase of $K$
  corresponds to the vector
  $\{(+,+),(+,-)\}$ and
  the $-$ phase
  corresponds to
  $\{(-,+),(-,-)\}$.
  For two configurations $\tau_V$ and $\tau_U$ with $V\cap U=\emptyset$,
  let $(\tau_V,\tau_U)$ be the joint configuration on $V\cup U$.
  Then we have the following:
  \begin{align}
	Z_K^\phase(\tau_{T(K)})
	& = \sum_{\tau_{C_1},\tau_{C_2},\tau_I}Z_{G_1}^\phase(\tau_{T(K)},\tau_{C_1})\left(Z_{G_2}^+(\tau_{C_2},\tau_{I})w(\tau_{C_1},\tau_{C_2})
	+ Z_{G_2}^-(\tau_{C_2},\tau_{I})w(\tau_{C_1},\tau_{C_2})\right)\notag \\
	& = \sum_{\tau_{C_1},\tau_{C_2}}Z_{G_1}^\phase(\tau_{T(K)},\tau_{C_1})w(\tau_{C_1},\tau_{C_2})
	\left(Z_{G_2}^+(\tau_{C_2})
	+Z_{G_2}^-(\tau_{C_2})\right),\label{eqn:ZK:tau}
  \end{align}
  where $w(\tau_{C_1},\tau_{C_2})$ denote the contribution from edges of $P$ given configurations $\tau_{C_1}$ and $\tau_{C_2}$,
  and $Z_K^\phase = \sum_{\tau_{T(K)}} Z_K^\phase(\tau_{T(K)})$.
  Moreover by Eq.~(\ref{eqn:near:ind}), for $i=1,2$ and any
  subset $S\subseteq T(G_i)$, we have
  \begin{align*}
	(1-\epsilon')Q^\phase(\tau_S)Z_{G_i}^\phase \leq Z_{G_i}^\phase(\tau_S) \leq (1+\epsilon')Q^\phase(\tau_S)Z_{G_i}^\phase,
  \end{align*}
  where we have used $Q^\phase(\tau_S)$ to denote the
  probability that the configuration on terminals in~$S$ is~$\tau_S$ in the distribution~$Q^\phase$.
  Therefore by Eq.~(\ref{eqn:ZK:tau}),
  \begin{align}
	Z_K^+ (\tau_{T(K)}) & \leq (1+\epsilon')^2 Q^+(\tau_{T(K)})
	Z_{G_1}^+Z_{G_2}^-
	\sum_{\tau_{C_1},\tau_{C_2}}Q^+(\tau_{C_1})w(\tau_{C_1},\tau_{C_2})
	\left(\frac{Z_{G_2}^+}{Z_{G_2}^-}Q^+(\tau_{C_2})
	+Q^-(\tau_{C_2})\right).\label{eqn:ZK+}
  \end{align}
  We need to calculate the quantity
  \begin{align}
	\label{eqn:mudef}
     \mu(\phase_1,\phase_2):=\sum_{\tau_{C_1},\tau_{C_2}}Q^{\phase_1}(\tau_{C_1})Q^{\phase_2}(\tau_{C_2})w(\tau_{C_1},\tau_{C_2}).
  \end{align}
  Recall our
  definitions
  of $M$ and $M^+$.
  Let $N=M^+ M (M^+)^{\texttt T}$ where ${\texttt T}$ means transposition.
  Then $\det(N)=(q^+-q^-)^2(\beta\gamma-1)<0$.
  We write
  \[N = \left(\begin{matrix}
          N_{--} & N_{-+}\\
          N_{+-} & N_{++}
        \end{matrix}\right)\]
  and let $c=\frac{N_{++}N_{--}}{N_{+-}N_{-+}}<1$.
  Here $N_{\phase_1\phase_2}$ is the edge contribution when one end point is chosen with probability $q^{\phase_1}$ and the other $q^{\phase_2}$.
  Also notice that each edge is independent under $Q^\pm$ so we can count them separately.
  Then the quantity in Eq.(\ref{eqn:mudef}) is
  \begin{align}
	\mu(+,+)=\mu(-,-)&=\left(N_{++}N_{--}\right)^{t'};\notag\\
	\mu(+,-)=\mu(-,+)&=\left(N_{+-}N_{-+}\right)^{t'}.
	\label{eqn:mu}
  \end{align}
  Plug Eq.~(\ref{eqn:mu}) in Eq.~(\ref{eqn:ZK+}),
  \begin{align}
	Z_K^+(\tau_{T(K)}) & \leq (1+\epsilon')^2 Z_{G_1}^+Z_{G_2}^-
	\left(\frac{Z_{G_2}^+}{Z_{G_2}^-}\left(N_{++}N_{--}\right)^{t'}
	+\left(N_{+-}N_{-+}\right)^{t'}\right)\cdot Q^+(\tau_{T(K)})\notag\\
	& = (1+\epsilon')^2 Z_{G_1}^+Z_{G_2}^- \left(N_{+-}N_{-+}\right)^{t'}
	\left(\frac{Z_{G_2}^+}{Z_{G_2}^-}c^{t'}+1\right)\cdot Q^+(\tau_{T(K)}).
	\label{eqn:ZK+:tau:upper}
  \end{align}
  Summing over $\tau_{T(K)}$ in Eq.~(\ref{eqn:ZK+:tau:upper}) we get
  \begin{align}
	Z_K^+ & \leq (1+\epsilon')^2 Z_{G_1}^+Z_{G_2}^- \left(N_{+-}N_{-+}\right)^{t'}
	\left(\frac{Z_{G_2}^+}{Z_{G_2}^-}c^{t'}+1\right).
	\label{eqn:ZK+:upper}
  \end{align}
  Similarly we get an estimate for $Z_K^-$:
  \begin{align}
	Z_K^- \geq (1-\epsilon')^2 Z_{G_1}^-Z_{G_2}^+ \left(N_{+-}N_{-+}\right)^{t'}
	\left(\frac{Z_{G_2}^-}{Z_{G_2}^+}c^{t'}+1\right).
	\label{eqn:ZK-:lower}
  \end{align}
  Let $r=\frac{Z_{G_2}^-}{Z_{G_2}^+}$.
  Notice that $Z_{G_1}^\phase = Z_{G_2}^\phase$ as $G_1$ and $G_2$ are identical copies.
  Combine Eq.~(\ref{eqn:ZK+:upper}) and Eq.~(\ref{eqn:ZK-:lower}),
  \begin{align}
	\frac{Z_K^-}{Z_K^+}\geq\left( \frac{1-\epsilon'}{1+\epsilon'} \right)^2\cdot\frac{1+c^{t'}r}{1+c^{t'}/r}.
	\label{eqn:ZK:ratio}
  \end{align}
  By Eq.~(\ref{eqn:rough:balance}) there is $f(t+t',\epsilon')$ such that
  \[\frac{1}{f(t+t',\epsilon')}\leq r\leq f(t+t',\epsilon'),\]
  and $f(t+t',\epsilon')$ is bounded above by a polynomial in $t+t'$ and $1/\epsilon'$.
  To show Eq.~(\ref{eqn:balance}) it suffices to show $\frac{Z_K^-}{Z_K^+}\geq\frac{1-\epsilon}{1+\epsilon}$ and $\frac{Z_K^+}{Z_K^-}\geq\frac{1-\epsilon}{1+\epsilon}$.
  Clearly $\frac{Z_K^-}{Z_K^+}\geq\frac{1-\epsilon}{1+\epsilon}$
  can be achieved by picking $\epsilon'=\frac{\epsilon}{3}$ and
  $t'=O(\log(t+\epsilon^{-1}))$
   in Eq.~(\ref{eqn:ZK:ratio}).
  To show $\frac{Z_K^+}{Z_K^-}\geq\frac{1-\epsilon}{1+\epsilon}$ is similar and therefore omitted.

Establishing that  Eq.~(\ref{eqn:near:ind}) still holds is easy to show.
  We will show
  \[\frac{Z_K^+(\tau_{T(K)})}{Z_K^+}\geq (1-\epsilon)Q^+(\tau_{T(K)}),\]
  and the other bounds are similar.
  By an argument similar to Eq.~(\ref{eqn:ZK+}), we have
  \[Z_K^+(\tau_{T(K)}) \geq (1-\epsilon')^2 Q^+(\tau_{T(K)}) Z_{G_1}^+Z_{G_2}^-
	\sum_{\tau_{C_1},\tau_{C_2}}Q^+(\tau_{C_1})w(\tau_{C_1},\tau_{C_2})\left(\frac{Z_{G_2}^+}{Z_{G_2}^-}Q^+(\tau_{C_2})
	+Q^-(\tau_{C_2})\right).\]
  Moreover, summing over $\tau_{T(K)}$ in Eq.~(\ref{eqn:ZK+}) we get:
  \[Z_K^+\leq (1+\epsilon')^2 Z_{G_1}^+Z_{G_2}^-
	\sum_{\tau_{C_1},\tau_{C_2}}Q^+(\tau_{C_1})w(\tau_{C_1},\tau_{C_2})\left(\frac{Z_{G_2}^+}{Z_{G_2}^-}Q^+(\tau_{C_2})
	+Q^-(\tau_{C_2})\right).\]
  The desired bound follows as $\epsilon'=\frac{\epsilon}{3}$.

  The other case is ferromagnetic, that is, $\beta\gamma>1$.
  We construct $K$ in the same way as in the previous case, but with the following change.
  To form the perfect matching $P$, we match $+/-$ terminals of $G_1$ to $-/+$ terminals of $G_2$.
  The proof goes similarly but $\det(N)=(q^+-q^-)^2(\beta\gamma-1)>0$.
  However since we made a twist in connecting $G_1$ and $G_2$,
  it follows that
  \begin{align*}
	\mu(+,+)=\mu(-,-)&=\left(N_{+-}N_{-+}\right)^{t'};\notag\\
	\mu(+,-)=\mu(-,+)&=\left(N_{--}N_{++}\right)^{t'}.
  \end{align*}
  Therefore we continue with the new constant $c'=\frac{N_{+-}N_{-+}}{N_{++}N_{--}}<1$ and the rest of the proof is the same.
\end{proof}

\section{Unary Symmetry Breaking}
\label{sec:sym}

In this section we prove Lemma \ref{lem:sym:break} that almost all 2-spin models support unary symmetry breaking.

\begin{proof}[Proof of Lemma \ref{lem:sym:break}]
Consider the sequence of gadgets $(H_k:k\in\mathbb{N})$, defined as follows.  The vertex set of
$H_k$ is $V(H_k)=\{u,u',u'',v_1,v_2,\ldots,v_k\}$, and $u$ is considered the attachment vertex.
The edge set of $H_k$ is
\[
E(H_k)=\big\{\{u',v_i\}:1\leq i\leq k\big\}\cup\big\{\{v_i,u''\}:1\leq i\leq k\big\}\cup\big\{\{u'',u\}\big\}.
\]
We shall argue that if the first three graphs in the sequence, namely $H_0$, $H_1$ and $H_2$,
all fail to be symmetry breakers then one of conditions (i) or (ii) holds.  The graph $H_0$ has an
isolated vertex that could clearly be removed;  we leave it in to make the calculations uniform.
Note that the maximum degree of any vertex in these graphs is~$3$.

Let $a=\beta^2+\lambda$, $b=\beta+\lambda\gamma$ and $c=1+\lambda\gamma^2$.  Then the
effective weight of vertex $u$ is given by the column vector,
\[
\rho^k=\begin{pmatrix}1&0\\ 0&\lambda\end{pmatrix}
\begin{pmatrix}\beta&1\\ 1&\gamma\end{pmatrix}
\begin{pmatrix}1&0\\ 0&\lambda\end{pmatrix}
\begin{pmatrix}a^k&b^k\\b^k&c^k\end{pmatrix}
\begin{pmatrix}1\\\lambda\end{pmatrix}
=
T\begin{pmatrix}
  a^k+\lambda b^k \\
  b^k+\lambda c^k
 \end{pmatrix}
\]
where
\[
T=\begin{pmatrix}\beta&\lambda\\ \lambda&\lambda^2\gamma\end{pmatrix}.
\]
For $h_k$ to be a symmetry breaker we require the vector $\rho^k$ not to be a multiple of $(1,\lambda)^{\texttt T}$.

Suppose $\rho^0$, $\rho^1$ and $\rho^2$ all fail to be symmetry breakers.  Then they must
all lie in a one-dimensional subspace of $\mathbb{R}^2$.
One way this can happen is if the matrix $T$ is rank~1,
i.e., if $\beta\gamma=1$.  This is case~(i).
Otherwise, $(1+\lambda,1+\lambda)^{\texttt T}$ and $(a+\lambda b,b+\lambda c)^{\texttt T}$
and $(a^2+\lambda b^2, b^2+\lambda c^2)^{\texttt T}$
lie in a one-dimensional subspace, namely the one generated by $(1,1)^{\texttt T}$.
This implies $a+\lambda b=b+\lambda c$ and $a^2+\lambda b^2=b^2+\lambda c^2$, or recasting,
\begin{align*}
a-b&=\lambda(c-b)\\
(a-b)(a+b)&=\lambda(c-b)(c+b).
\end{align*}
So either $a=b=c$, or (dividing the second equation by the first)
$a=c$ and $\lambda=1$.   In either case, substituting for $a$, $b$ and $c$ in terms
of $\beta$, $\gamma$, $\lambda$, we obtain either $\beta=\gamma=1$ (which belongs to case (i)) or $\beta=\gamma$ and $\lambda=1$ (which is case (ii)).
\end{proof}

There are two exceptional cases.
The first is $\beta\gamma=1$.
It is well-known that in this case the 2-spin system can be decomposed and hence tractable.
The other case of $\beta=\gamma$ and $\lambda=1$ is perfectly symmetric and this symmetry cannot be broken.
This system is the Ising model without external fields.
For this system, the marginal probability of any vertex in any graph is exactly $1/2$.

Regarding the ferromagnetic case, Jerrum and Sinclair \cite{JS93} presented an FPRAS for the ferromagnetic Ising model with consistent external fields.
On the other hand, anti-ferromagnetic Ising models without external field on bipartite graphs are actually
equivalent to ferromagnetic Ising models.
The trick is to flip the assignment on only one
part of the
bipartition, which has been used before by Goldberg and Jerrum \cite{GJ07}.

\begin{lemma}
  For $0<\alpha<1$, \BSpin{\alpha}{\alpha}{1}~$\equiv_{\texttt T}$~\BSpin{1/\alpha}{1/\alpha}{1}.
  \label{lem:Ising:w/o}
\end{lemma}
\begin{proof}
  Let $B=(V_1,V_2,E)$ be a bipartite graph where $V_1$ and $V_2$ are two partitions of vertices.
  Let $|E|=m$.
  Then we have
  \begin{align*}
	Z_{B}\left( \alpha,\alpha,1\right)
	& = \sum_{\sigma_{V_1}:V_1\rightarrow\{0,1\}}\sum_{\sigma_{V_2}:V_2\rightarrow\{0,1\}}
	\prod_{(v,w)\in E}\alpha^{(1-\sigma_{V_1}(v))(1-\sigma_{V_2}(w))}\alpha^{\sigma_{V_1}(v)\sigma_{V_2}(w)}\\
	& = \sum_{\sigma_{V_1}:V_1\rightarrow\{0,1\}}\sum_{\sigma_{V_2}:V_2\rightarrow\{0,1\}}
	\prod_{(v,w)\in E}\alpha^{(1-\sigma_{V_1}(v))\sigma_{V_2}(w)}\alpha^{\sigma_{V_1}(v)(1-\sigma_{V_2}(w))}\\
	& = \sum_{\sigma_{V_1}:V_1\rightarrow\{0,1\}}\sum_{\sigma_{V_2}:V_2\rightarrow\{0,1\}}
	\prod_{(v,w)\in E}\alpha^{\sigma_{V_2}(w)-\sigma_{V_1}(v)\sigma_{V_2}(w)+\sigma_{V_1}(v)-\sigma_{V_1}(v)\sigma_{V_2}(w))}\\
	& = \alpha^m\sum_{\sigma_{V_1}:V_1\rightarrow\{0,1\}}\sum_{\sigma_{V_2}:V_2\rightarrow\{0,1\}}
	\prod_{(v,w)\in E}\alpha^{-(1-\sigma_{V_1}(v))(1-\sigma_{V_2}(w))-\sigma_{V_1}(v)\sigma_{V_2}(w)}\\	
	& = \alpha^m Z_{B}\left( \alpha^{-1},\alpha^{-1},1\right),
  \end{align*}
  where in the second line we flip the assignment of $\sigma_{V_2}$.
\end{proof}

Combining Lemma~\ref{lem:Ising:w/o} with the FPRAS by Jerrum and Sinclair \cite{JS93} for the ferromagnetic Ising model
yields the following corollary.

\begin{corollary}
  For any $\alpha>0$, \BSpin{\alpha}{\alpha}{1}  has an FPRAS.
  \label{cor:Ising:w/o}
\end{corollary}

As explained in the introduction, this
corollary helps explain why the notion of unary symmetry breaking is necessary to achieve \BIS-hardness.

\section{Nearly-independent Phases for Antiferromagnetic Systems}
\label{sec:indpt-phases}

In this section we prove Lemma \ref{lem:indpt-phases} that for 2-spin
antiferromagnetic systems in the tree non-uniqueness region there is
a gadget with nearly-independent phase-correlated spins.

We first give some necessary background which details how the values $q^-$ and $q^+$ are derived.

\subsection{Gibbs Measures on Infinite Trees}

Recall that the Gibbs distribution is
the distribution in which a configuration~$\sigma$ is drawn
with probability
\begin{align}
  \Pr_{G;\beta,\gamma,\lambda}(\sigma)=\frac{w(\sigma)}{Z_G(\beta,\gamma,\lambda)}.
  \label{eqn:Gibbs}
\end{align}
Let $\inftree{\Delta}$ be the infinite $(\Delta-1)$-ary tree.
A \emph{Gibbs measure} on $\inftree{\Delta}$ is a measure
such that for any finite subtree $T\subset\inftree{\Delta}$,
the induced distribution on $T$ conditioned on the outer boundary is the same as that given by (\ref{eqn:Gibbs}).
There may be one or more Gibbs measures (see, e.g.,~\cite{Geo11} for more details).
A Gibbs measure is called \emph{translation-invariant} if it is invariant under all automorphisms of $\inftree{\Delta}$,
and is \emph{semi-translation-invariant} if it is invariant under all parity-preserving automorphisms of $\inftree{\Delta}$.
In our context, the Gibbs measures that will be of interest  are the two extremal semi-translation-invariant Gibbs measures corresponding to the all $1$'s and all $0$'s boundary conditions. These two measures are different in the non-uniqueness region of $\inftree{\Delta}$. 
Let $0<\qm<\qp<1$ be the two marginal probabilities of the root
having spin~$1$ in the two extremal semi-translation-invariant Gibbs measures. More precisely, one can define $q^+,q^-$ as follows.
For $s\in\{0,1\}$, let $q_{\ell,s}$ denote the probability that the root is assigned spin~$1$ in the $(\Delta-1)$-ary tree of depth $\ell$
in the Gibbs distribution where the leaves are fixed to spin $s$. In standard terminology, fixing the configuration on the leaves to all 1's or all 0's is most commonly referred to as the $+,-$ boundary conditions, respectively (and hence the indices $+,-$ in our notation of $q^+,q^-$). It is not hard to show that $q_{2\ell,1}$ is decreasing in $\ell$, while $q_{2\ell,0}$ is increasing. Let $q^+ = \lim_{\ell\rightarrow\infty} q_{2\ell,1}$ and let
$q^- = \lim_{\ell\rightarrow\infty} q_{2\ell,0}$. 
The two quantities $\qp$ and $\qm$ satisfy the standard tree recursion in the following sense.
Let $r^\phase=\frac{q^\phase}{1-q^\phase}$ for $\phase\in\phases$, and $f(x)=\lambda\left( \frac{1+\beta x}{\gamma+x} \right)^{\Delta-1}$.
Then
\begin{align*}
  r^+=f(r^-),~{\rm and}~r^-=f(r^+).
\end{align*}

\subsection{Sly's Gadget}\label{sec:slygadget}

Sly (\cite{Sly10}, Theorem 2.1)~showed that for every $\Delta\geq 3$, there exists $\epsilon_{\Delta}$
such that the hard-core model $(1,0,\lambda,\Delta)$ supports nearly-independent phase-correlated spins
for any $\lambda$ satisfying $\lambda_c(\inftree{\Delta})<\lambda<\lambda_c(\inftree{\Delta})+\epsilon_{\Delta}$.
This region is a small interval just above the uniqueness threshold.
In the same paper Sly also showed that $(1,0,1,6)$ supports nearly-independent phase-correlated spins.
The quantities defining measures $Q^\pm$ in Definition~\ref{def:sly} are exactly the marginal probabilities $q^{\pm}$ on $\mathbb{T}_\Delta$.

Galanis et al.~\cite{GGSVY11} extended the applicable region of Sly's gadget for the hard-core model
to all $\lambda>\lambda_c(\mathbb{T}_\Delta)$ for $\Delta=3$ and $\Delta\geq 6$.
The gap of $\Delta=4$ and $\Delta=5$ is later closed \cite{GSV12}.
In the later paper Galanis et al.~\cite{GSV12} also verified Sly's gadget for parameters $(\beta,\gamma,\lambda,\Delta)$ such that $\beta\gamma<1$,
$\sqrt{\beta\gamma}\geq\frac{\sqrt{\Delta-1}-1}{\sqrt{\Delta-1}+1}$,
and $(\beta,\gamma,\lambda)$ is in the non-uniqueness region of infinite tree $\mathbb{T}_\Delta$.
Using \cite{GSV13} we extend the applicable region of Sly's gadget
to the entire non-uniqueness region for all 2-spin antiferromagnetic models.

Sly's construction is based on a technical analysis of the configurations with the largest contribution to the partition function of a random bipartite $\Delta$-regular graph. This required a second moment argument involving a optimization which posed technical difficulties. Following recent developments in \cite{GSV13}, we describe how to obtain the properties of Sly's gadget for all 2-spin antiferromagnetic systems (including the case with an external field) in the non-uniqueness region.   We will give an overview of the second moment approach and how to tackle it using the technique in \cite{GSV13}.  We begin by
presenting the relevant definitions and  describing the construction of the gadget.

For integers $r,n$, let $\Gc^r_{n}$ be the following graph distribution:
\begin{enumerate}
\item $\Gc^r_{n}$ is supported on bipartite graphs. The two parts of the bipartite graph are labelled by $+,-$ and each is partitioned as $V^{\pi}:=U^{\pi}\cup W^{\pi}$ where $|U^{\pi}|=n$, $|W^{\pi}|=r$ for $\pi=\{+,-\}$. $U$ denotes the set $U^{+}\cup U^{-}$ and similarly $W$ denotes the set $W^{+}\cup W^{-}$.
\item To sample $G\sim\Gc^r_{n}$,  sample uniformly and independently $\Delta$ matchings: (i) $(\Delta-1)$ random perfect matchings between $U^+\cup W^+$ and $U^-\cup W^-$,  (ii) a random perfect matching between $U^+$ and $U^-$. The edge set of $G$ is the union of the $\Delta$ matchings. Thus, vertices in $U$ have degree $\Delta$, while vertices in $W$ have degree $\Delta-1$.
\end{enumerate}
The case $r=0$ will also be important, in which case we denote the distribution as $\Gc_n$. Note that $\Gc_n$ is supported on bipartite $\Delta$-regular graphs. Strictly speaking, $\Gc^r_{n}, \Gc_n$ are supported on multi-graphs, but it is well known that every statement that holds asymptotically almost surely on these spaces continues to hold asymptotically almost surely conditioned on the event that the graph is simple.

We are now ready to give the construction of the gadget. For constants $0<\theta,\psi<1/8$, let $m'=(\Delta-1)^{\lfloor \theta \log_{\Delta-1}n\rfloor+2\lfloor \frac{\psi}{2} \log_{\Delta-1}n\rfloor}$.  Note that $m'=o(n^{1/4})$. First, sample $G$ from the distribution $\Gc^{m'}_{n}$ conditioning on $G$ being simple. Next, for $\pi\in\{+,-\}$, attach $t$ disjoint $(\Delta-1)$-ary trees of even depth $\ell$ (with $t=(\Delta-1)^{\lfloor \theta \log_{\Delta-1}n\rfloor}$ and $\ell=2\lfloor \frac{\psi}{2} \log_{\Delta-1}n\rfloor$) to $W^\pi$, so that every vertex in $W^\pi$  is a leaf of exactly one tree (this is possible since $m'=|W|=t (\Delta-1)^\ell$). Denote by $T^\pi$ the roots of the trees, so that $|T^{\pi}|=t$. The trees do not share common vertices with the graph $G$, apart from the vertices in $W$. The final graph $\widetilde{G}$ is the desired gadget, where the terminals $T$ are the roots of the trees. We denote the family of graphs that can be constructed this way by $\widetilde{\calG}(t,n(t),\Delta)$. Note that the size of the construction is $(2+o(1))n$ which is bounded above by a polynomial in $t$ when $\Delta$ is a fixed constant. Moreover, any $\widetilde{G}$ drawn from $\widetilde{\calG}(t,n(t),\Delta)$ is bipartite. The terminals of $\widetilde{G}$ are $T(G)=T^+\cup T^-$, and $T^+$ and $T^-$ are on distinct partitions of the bipartite graph.

Next we show that if $\beta,\gamma,\lambda$ lie in the non-uniqueness regime of $\TreeD$,
then the gadget $\widetilde{G}$ satisfies Definition \ref{def:sly}.
Lemma~19 in \cite{GSV12} shows that this
is true, assuming that a certain technical condition is true \cite[Condition 1]{GSV12}.
We will introduce this condition and then show that it always holds when $\beta,\gamma,\lambda$ lie in the non-uniqueness regime of $\TreeD$.

The technical condition involves asymptotics of the leading terms of the first and second moments.
We will derive expressions for these shortly.
Moreover, the condition is stated for the case when $r=0$, as $r=o(n^{1/4})$ is relatively small and 
\cite{GSV12} shows that it does not affect the leading terms.
Suppose that $G\sim \Gc_n$.
We look at the contribution of configurations which assign a prescribed fraction of vertices in $U^+,U^-$ the spin 1,
with a view to identifying the phases of the gadget. 
Here $V^+=U^+$ and $V^-=U^-$ and $n=|V^+|=|V^-|$. For $1\geq \chi^+,\chi^{-}\geq0$, let
\[\Sigma^{\chi^+,\chi^-}:=\{\sigma: V^+\cup V^-\rightarrow\{0,1\}\, |\, |\sigma^{-1}(1)\cap U^+|=\ceil{\chi^+ n},\, |\sigma^{-1}(1)\cap U^-|=\ceil{\chi^- n},\}\]
where $\sigma^{-1}(1)$ denotes the set of vertices assigned the spin $1$ in the configuration $\sigma$.  
Denote by $Z_{G}^{\chi^+,\chi^-}$ the total contribution to the partition function of $G$ by the configurations in $\Sigma^{\chi^+,\chi^-}$, 
that is,
\[Z_{G}^{\chi^+,\chi^-}=\sum_{\sigma\in\Sigma_r^{\chi^+,\chi^-}} w_G(\sigma).\]
We will calculate the leading term in the exponent (as a function of $n$) of the first and second moment of $Z_{G}^{\chi^+,\chi^-}$.

By linearity we have $\E_{\Gc_n}\big[Z^{\chi^+,\chi^-}_G\big]=\sum_{\sigma\in \Sigma^{\chi^{+},\chi^{-}}}\E_{\Gc_n}\big[w_G(\sigma)\big]$.
Fix $\sigma\in\Sigma^{\chi^{+},\chi^{-}}$. 
For convenience, set $\chi^+_n=\left\lceil n \chi^+\right\rceil/n, \chi^-_n=\left\lceil n \chi^-\right\rceil/n$. 
Let $nx_n$ be the number of $(1,1)$ edges in a random perfect matching between $V^+,V^-$ under the configuration $\sigma$. 
Note that $\chi_n\leq x_n\leq \chi_n'$, where $\chi_n=\max\{0,\chi^+_n+\chi^-_n-1\}$ and $\chi'_n=\min\{\chi^{+}_n,\chi^{-}_n\}$.  
It is not hard to see that
\begin{equation}\label{one2april}
\E_{\Gc_n}\big[w_G(\sigma)\big]=\lambda^{n(\chi^+_n + \chi^-_n)}
 \left(\sum_{\chi_n\leq x_n\leq\chi_n'}
\kappa^{\chi^{+}_n,\chi^{-}_n}_{x_n} \right)^{\Delta}, \mbox{ where }
\kappa^{\chi^{+}_n,\chi^{-}_n}_{x_n} =  
\frac
{\binom{n\chi^+_n}{nx_n}
\binom{n(1-\chi^+_n)}{n(\chi^-_n-x_n)}}
 {\binom{n}{n\chi^-_n}}
\beta^{ n(1-\chi^+_n-\chi^-_n+x_n)}
 \gamma^{nx_n}.
\end{equation}
By the linearity of expectation and symmetry, we then have
\begin{align}
\E_{\Gc_n}\big[
Z^{\chi^{+},\chi^{-}}_{G}\big] =\big|\Sigma^{\chi^+,\chi^-}\big|\, \E_{\Gc_n}\big[w_G(\sigma)\big], \mbox{ where } \big|\Sigma^{\chi^+,\chi^-}\big|=\binom{n}{n\chi^+_n}\binom{n}{n\chi^-_n}.
\label{two2april}
\end{align} 

Let $\chi:=\max\{0,\chi^++\chi^--1\}$ and $\chi':=\min\{\chi^+,\chi^-\}$. Note that as $n\rightarrow\infty$, $(\chi^+_n,\,\chi^-_n,\,\chi_n,\,\chi_n')\rightarrow (\chi^+,\, \chi^-,\, \chi,\, \chi')$. In Section~\ref{sec:asymptotics}, using Stirling's approximation, it is  shown that
\begin{equation}
\label{ssone}
\forall \chi^+,\chi^-,\quad
  \frac{1}{n}\log\E_{\Gc_n}\big[Z^{\chi^+,\chi^-}_G\big]=\Psi_{1;\beta,\gamma,\lambda}(\chi^+,\chi^-)+o(1),
\end{equation}
where, under the convention $0\log 0\equiv0$,
\begin{equation}\label{eq:firstmoment}
\begin{aligned}
  \Psi_{1;\beta,\gamma,\lambda}(\chi^+,\chi^-)&:=\max_{\chi\leq x\leq \chi'}\Psi'_{1;\beta,\gamma,\lambda}(x,\chi^+,\chi^-),\\
  \Psi'_{1;\beta,\gamma,\lambda}(x,\chi^+,\chi^-)&:=(\chi^++\chi^-)\log \lambda+(\Delta-1)f_1(\chi^+,\chi^-)+\Delta g_{1;\beta,\gamma}(x,\chi^+,\chi^-),\\
  f_1(\chi^+,\chi^-)&:=\chi^+\log \chi^+ +(1-\chi^+)\log(1-\chi^+) +\chi^-\log \chi^-+(1-\chi^-)\log(1-\chi^-),\\
  g_{1;\beta,\gamma}(x,\chi^+,\chi^-)&:=(1-\chi^+-\chi^-+x)\log\beta+x\log\gamma\\
  &-x\log x-(\chi^+-x)\log(\chi^+-x)-(\chi^--x)\log(\chi^--x)\\
  &-(1-\chi^+-\chi^-+x)\log(1-\chi^+-\chi^-+x).
\end{aligned}
\end{equation}
Note that all the quantities in \eqref{eq:firstmoment} are independent of $n$.

To calculate the second moment, the approach is completely analogous though we need to introduce more variables. For $1\geq \upsilon^+, \upsilon^-\geq 0$, let
\begin{equation}\label{eq:sigma2def}
  \Sigma^{\chi^+,\chi^-}_{\upsilon^+,\upsilon^-}
  :=\left\{(\sigma_1,\sigma_2)\mid\sigma_1,\sigma_2\in\Sigma^{\chi^+,\chi^-},
  |\sigma_1^{-1}(1)\cap\sigma_2^{-1}(1)\cap U^+|=\ceil{\upsilon^+ n},
  |\sigma_1^{-1}(1)\cap\sigma_2^{-1}(1)\cap U^-|=\ceil{\upsilon^- n}\right\}.
\end{equation}
Let $Y^{\chi^+,\chi^-}_{\upsilon^+,\upsilon^-}=\sum_{(\sigma_1,\sigma_2)\in\Sigma^{\chi^+,\chi^-}_{\upsilon^+,\upsilon^-}}w_G(\sigma_1)w_G(\sigma_2)$.
Notice that $Y^{\chi^+,\chi^-}_{\upsilon^+,\upsilon^-}$ is the contribution to $(Z^{\chi^+,\chi^-}_G)^2$
from pairs of configurations in $\Sigma^{\chi^+,\chi^-}$ with overlap $\ceil{\upsilon^+n}$ and 
$\ceil{\upsilon^-n}$. Then, we have
\begin{equation}\label{eq:sum2nd}
  \E_{\Gc_n}\left[(Z_G^{\chi^+,\chi^-})^2\right]
  =\sum_{\upsilon^+,\upsilon^-}\E_{\Gc_n}\left[Y^{\chi^+,\chi^-}_{\upsilon^+,\upsilon^-}\right],
\end{equation}
where the sum is over integer multiples $\upsilon^+$ and $\upsilon^-$ of~$1/n$. 
Later we will replace this sum with a maximum and consider general real numbers~$\upsilon^+$
and~$\upsilon^-$. To unify the notation in both cases,
set $\upsilon^+_n=\ceil{n \upsilon^+}/n$ and $\upsilon^-_n=\ceil{n \upsilon^-}/n$. 
Note that 
the overlap $\upsilon^+_n$ is at most $\chi^+_n n$
and at least $2 \chi^+_n n-n$ and similarly for the minus side.

To get an expression for $\E_{\Gc_n}\left[Y^{\chi^+,\chi^-}_{\upsilon^+,\upsilon^-}\right]$, fix arbitrary $(\sigma_1,\sigma_2)\in \Sigma^{\chi^+,\chi^-}_{\upsilon^+,\upsilon^-}$. The configurations $\sigma_1$ and $\sigma_2$ divide
each of~$U^{+}$ and~$U^-$ into $4$ subsets.
Let the subsets chosen for spin~$1$ by both~$\sigma_1$ and~$\sigma_2$ be numbered~$1$ on
the~$+$ side and the~$-$ side.
Similarly, let the subsets chosen for spin~$1$ in~$\sigma_1$ and for spin~$0$ in~$\sigma_2$ be numbered~$2$
and the subsets chosen for spin~$0$ in~$\sigma_1$ and spin~$1$ in~$\sigma_2$ be numbered~$3$.
Finally, let the subsets chosen for spin~$0$ in both~$\sigma_1$ and~$\sigma_2$ be numbered~$4$.
 Let $ny_{ijn}$ be the number of  edges between subset $i$ on one side and subset $j$ on the other, and $\mathbf{y}_n=\{y_{ijn}\}$.
Define
\begin{align*}
  L_{1n}=\upsilon^+_n,\quad L_{2n}=L_{3n}=\chi^+_n-\upsilon^+_n,\quad 
  L_{4n}=1-2\chi^+_n+\upsilon^+_n,\\
  R_{1n}=\upsilon^-_n,\quad R_{2n}=R_{3n}= \chi^-_n-\upsilon^-_n,\quad 
  R_{4n}=1-2\chi^-_n+\upsilon^-_n.\\
 \end{align*}
Observe that the $\mathbf{y}_n$ satisfy
\begin{equation}\label{eqn:2nd:constraints}
\begin{gathered}
  \sum_j y_{ijn}=L_{in}~\text{ for }i\in[4],\quad \sum_i y_{ijn}=R_{jn},~\text{ for }j\in[4],\\
  y_{ijn}\geq 0~\text{ and } y_{ijn} \text{ is an integral  multiple of $1/n$}.
\end{gathered}
\end{equation}
Then we have
\begin{equation}
\E_{\Gc_n}[w(\sigma_1)w(\sigma_2)]=\lambda^{2n(\chi^+_n + \chi^-_n)}\bigg( \sum_{\mathbf{y_n}}\xi^{\chi^+_n, \chi^-_n}_{\upsilon^+_n,\upsilon^-_n,\mathbf{y}_n}\bigg)^{\Delta},\label{longeq}
\end{equation}
where 
\begin{equation}\label{eq:defxi}
\begin{aligned}
\xi^{\chi^+_n, \chi^-_n}_{\upsilon^+_n,\upsilon^-_n,\mathbf{y}_n}&:={n\choose ny_{11n},ny_{12n},\dots,ny_{44n}}^{-1}\prod_{i=1}^4{nL_{in} \choose ny_{i1n},ny_{i2n},ny_{i3n},ny_{i4n}} 
  \prod_{j=1}^4{nR_{jn} \choose ny_{1jn},ny_{2jn},ny_{3jn},ny_{4jn}}\\
&\times \beta^{n(y_{22n}+y_{24n}+y_{33n}+y_{34n}+y_{42n}+y_{43n}+2y_{44n})}\gamma^{(2y_{11n}+y_{12n}+y_{13n}+y_{21n}+y_{31n}+y_{22n}+y_{33n})n}
\end{aligned}
\end{equation}
To see why \eqref{longeq} is true note that a random matching may be obtained by fixing the permutation within each of the
$16$ subsets $y_{ij}$ of the ``$-$'' side and then choosing the $16$ corresponding blocks of the ``$+$'' side giving the
denominator $n!={n\choose ny_{11n},ny_{12n},\dots,ny_{44n}}
\prod_{i,j} (ny_{ijn})!$.
Then the numerator counts all matchings consistent with~$\sigma_1$ and~$\sigma_2$ giving
$\prod_{i=1}^4{nL_{in} \choose ny_{i1n},ny_{i2n},ny_{i3n},ny_{i4n}} 
  \prod_{j=1}^4{nR_{jn} \choose ny_{1jn},ny_{2jn},ny_{3jn},ny_{4jn}}
\prod_{i,j} (ny_{ijn})!$. The factors $\prod_{i,j} (ny_{ijn})!$ cancel, giving the first line in \eqref{eq:defxi}. The second line is just $w(\sigma_1) w(\sigma_2)$.

To account for the cardinality of $\Sigma^{\chi^+,\chi^-}_{\upsilon^+,\upsilon^-}$, by the linearity of expectation and symmetry, we obtain
\begin{equation}\label{longeqb}
  \E_{\Gc_n}\left[ Y^{\chi^+,\chi^-}_{\upsilon^+,\upsilon^-} \right]
  ={n\choose nL_{1n},nL_{2n},nL_{3n},nL_{4n}}{n\choose nR_{1n},nR_{2n},nR_{3n},nR_{4n}}\, \E_{\Gc_n}[w(\sigma_1)w(\sigma_2)]
\end{equation}

Note that as $n\rightarrow \infty$, $L_{in}\rightarrow L_{i},\ R_{jn}\rightarrow R_j\ (\forall i,j\in[4]),$ where
\begin{equation}\label{eq:LRlimits}
\begin{gathered}
  L_1=\upsilon^+,\quad L_2=L_3=\chi^+-\upsilon^+,\quad L_4=1-2\chi^++\upsilon^+,\\
  R_1=\upsilon^-,\quad R_2=R_3=\chi^--\upsilon^-,\quad R_4=1-2\chi^-+\upsilon^-.
\end{gathered}
\end{equation}  
 
Using again Stirling's approximation, in Section~\ref{sec:asymptotics} it is shown that 
\begin{align}
\label{lastone}
  \frac{1}{n}\log\E_{\Gc_n}\left[ Y^{\chi^+,\chi^-}_{\upsilon^+,\upsilon^-} \right] =\Psi_{2;\beta,\gamma,\lambda}'(\chi^+,\chi^-,\upsilon^+,\upsilon^-)+o(1),
\end{align}
where
\begin{align*} \Psi_{2;\beta,\gamma,\lambda}'(\chi^+,\chi^-,\upsilon^+,\upsilon^-)&:=\max_{\mathbf{y}}\Psi_{2;\beta,\gamma,\lambda}''(\chi^+,\chi^-,\upsilon^+,\upsilon^-,\mathbf{y});\\
  \Psi_{2;\beta,\gamma,\lambda}''(\chi^+,\chi^-,\upsilon^+,\upsilon^-,\mathbf{y})&:=2(\chi^+ + \chi^-)\log\lambda+(\Delta-1)f_2(\chi^+,\chi^-,\upsilon^+,\upsilon^-)+\Delta g_{2;\beta,\gamma,\lambda}(\mathbf{y});\\
  f_2(\chi^+,\chi^-,\upsilon^+,\upsilon^-)&:=2(\chi^+-\upsilon^+)\log(\chi^+-\upsilon^+)+\upsilon^+\log\upsilon^++(1-2\chi^++\upsilon^+)\log(1-2\chi^++\upsilon^+)\\
  &+2(\chi^--\upsilon^-)\log(\chi^--\upsilon^-)+\upsilon^-\log\upsilon^-+(1-2\chi^-+\upsilon^-)\log(1-2\chi^-+\upsilon^-);\\
  g_{2;\beta,\gamma}(\mathbf{y})&:=(y_{22}+y_{24}+y_{33}+y_{34}+y_{42}+y_{43}+2y_{44})\log\beta\\
  &+(2y_{11}+y_{12}+y_{13}+y_{21}+y_{22}+y_{31}+y_{33})\log\gamma-\sum_{ij}y_{ij}\log y_{ij},
\end{align*}
where the maximum is over non-negative (real vectors)  $\mathbf{y}$ satisfying 
\begin{equation}\label{eqn:newconstraints}
\sum_j y_{ij}=L_ i,~\text{ for }i\in[4], \quad \sum_i y_{ij}=R_j,~\text{ for }j\in[4].\\
\end{equation}

To obtain the leading term of the second moment $\E_{\Gc_n}\left[(Z_G^{\chi^+,\chi^-})^2\right]$, 
define $$
\mathcal{D} := \Big\{ (\upsilon^+,\upsilon^-) \mid
\max\{0,2 \chi^+-1\} \leq \upsilon^+ \leq \chi^+
\text{ and }
\max\{0,2 \chi^--1\} \leq \upsilon^+ \leq \chi^-
\Big\}.
$$ 
 In Section~\ref{sec:asymptotics}, we  will show
\begin{equation}\label{sstwo}
\forall \chi^+,\chi^-, \quad
  \frac{1}{n}\log\E_{\Gc_n}\left[(Z_G^{\chi^+,\chi^-})^2\right]
  =\Psi_{2;\beta,\gamma,\lambda}(\chi^+,\chi^-)+o(1), 
\end{equation}
where
\begin{equation*}
  \Psi_{2;\beta,\gamma,\lambda}(\chi^+,\chi^-):=\max_{(\upsilon^+,\upsilon^-)\in\mathcal{D}} \Psi_{2;\beta,\gamma,\lambda}'(\chi^+,\chi^-,\upsilon^+,\upsilon^-).
\end{equation*}

We want to focus on $(\chi^+,\chi^-)$ which maximize $\Psi_{1;\beta,\gamma,\lambda}$. It is shown in \cite[Lemma 1]{GSV12} that when non-uniqueness holds, there exist $p^+,p^-$ with $p^+>p^-$ such that  $(\chi^+,\chi^-)=(p^+,p^-)$  is the unique maximizer (up to symmetry) of $\Psi_{1;\beta,\gamma,\lambda}$. We should note here that the values $p^+,p^-$ are the analogues of $q^+,q^-$ for the infinite $\Delta$-regular tree (observe that the infinite $\Delta$-regular tree differs from $\TreeD$ only at the degree of the root).
The technical condition  from \cite{GSV12} is as follows.

\begin{condition}[\protect{\cite[Condition 1]{GSV12}}]\label{GSV12:condition1}\label{cond:cond}
  $\Psi_{2;\beta,\gamma,\lambda}'(p^+,p^-,\upsilon^+,\upsilon^-)$ is maximized at $\upsilon^+=(p^+)^2$, $\upsilon^-=(p^-)^2$.
\end{condition}

To see that our version of Condition~\ref{cond:cond} is equivalent to the version in \cite{GSV12}, use (\ref{lastone}).
Our interest in Condition~\ref{cond:cond} is justified by the following lemma which is proved in \cite{GSV12}.

\begin{lemma}[\protect{\cite[Lemma 19]{GSV12}}]\label{lem:GSV12:condition1}
  When the parameter set $(\beta,\gamma,\lambda)$ lies in the non-uniqueness region of $\TreeD$,
  and Condition \ref{GSV12:condition1} holds,
  the gadget $\widetilde{G}$ satisfies Definition \ref{def:sly}.
\end{lemma}

We will establish Condition \ref{GSV12:condition1} when $(\beta,\gamma,\lambda)$ lies in the non-uniqueness region of $\TreeD$.
To do so, we apply results from \cite{GSV13}.
The next lemma is a specialization of Theorem $3$ from \cite{GSV13} to our setting
for the case that there is no field.

\begin{lemma}[\protect{\cite[Theorem 3]{GSV13}}]  \label{lem:GSV13:thm3}
  Suppose  $\Delta\geq 3$. 
   Then $$\max_{\chi^+,\chi^-}\Psi_{2;\beta,\gamma,1}(\chi^+,\chi^-)=2\max_{\chi^+,\chi^-} \Psi_{1;\beta,\gamma,1}(\chi^+,\chi^-)$$.
\end{lemma}

The main difference between the setting in \cite{GSV13} and ours
 is that we allow external fields whereas they do not.
 We use the following lemma
to apply Lemma~\ref{lem:GSV13:thm3}
to our setting.
 
 \begin{lemma}
 \label{lem:new}
 For all  $\beta$, $\gamma$, $\lambda$, $\Delta$,  $\chi^+$, $\chi^-$,
$\upsilon^+$ and $\upsilon^-$,
 \begin{align*}
\Psi_{1;\beta,\gamma,\lambda}(\chi^+,\chi^-) &= \Psi_{1; \beta/ \lambda^{1/\Delta}, \gamma \lambda^{1/\Delta},1}(\chi^+,\chi^-)+\log\lambda,\\
\Psi_{2;\beta,\gamma,\lambda}'(\chi^+,\chi^-,\upsilon^+,\upsilon^-) 
&= \Psi_{2; \beta/ \lambda^{1/\Delta}, \gamma \lambda^{1/\Delta},1}'(\chi^+,\chi^-,\upsilon^+,\upsilon^-)+2\log\lambda,\\
\Psi_{2;\beta,\gamma,\lambda}(\chi^+,\chi^-) &= \Psi_{2; \beta/ \lambda^{1/\Delta}, \gamma \lambda^{1/\Delta},1}(\chi^+,\chi^-)+2\log\lambda\\
\end{align*}
 \end{lemma}
 \begin{proof}
For the first equation, note from the definition of $g_1$ that 
\begin{align*}
  g_{1; \beta,\gamma} (x,\chi^+,\chi^-) 
  - g_{1; \beta/ \lambda^{1/\Delta}, \gamma \lambda^{1/\Delta}} (x,\chi^+,\chi^-)
  = \log(\lambda)(1-\chi^+-\chi^-)/\Delta.
\end{align*}
Thus, from the definition of $\Psi_1$,
we have
\begin{align*}
  \Psi_{1;\beta,\gamma,\lambda}(\chi^+,\chi^-)&-\Psi_{1; \beta/ \lambda^{1/\Delta}, \gamma \lambda^{1/\Delta},1}(\chi^+,\chi^-)
  =\log\lambda.
\end{align*}
Similarly, from the definition of $\Psi_2''$, we have
\begin{align*}
  \Psi_{2;\beta,\gamma,\lambda}''(\chi^+,\chi^-,\upsilon^+,\upsilon^-,\mathbf{y})
  -\Psi_{2;\beta/ \lambda^{1/\Delta}, \gamma \lambda^{1/\Delta},1}''(\chi^+,\chi^-,\upsilon^+,\upsilon^-,\mathbf{y})
 & = 2(\chi^++\chi^-)\log\lambda\\
 &+\Delta(g_{2;\beta,\gamma}(\mathbf{y})-g_{2;\beta/ \lambda^{1/\Delta}, \gamma \lambda^{1/\Delta}}(\mathbf{y})),
\end{align*}
and from the definition of $g_2$, 
\begin{align*}
  \Delta(g_{2;\beta,\gamma}(\mathbf{y})
  -g_{2;\beta/ \lambda^{1/\Delta}, \gamma \lambda^{1/\Delta}}(\mathbf{y}))
  & =(y_{22}+y_{24}+y_{33}+y_{34}+y_{42}+y_{43}+2y_{44}\\ 
  &-2y_{11}-y_{12}-y_{13}-y_{21}-y_{22}-y_{31}-y_{33})\log\lambda.
\end{align*}
However by \eqref{eqn:newconstraints}
\begin{align*} 
  (y_{22}+y_{24}+y_{33}+y_{34}+y_{42}+y_{43}+2y_{44}-2y_{11}-y_{12}-y_{13}-y_{21}-y_{22}-y_{31}-y_{33})
  &=L_4+R_4-L_1-R_1\\
  &=2-2(\chi^++\chi^-).
\end{align*}
Therefore we have
\begin{align*}
  \Psi_{2;\beta,\gamma,\lambda}''(\chi^+,\chi^-,\upsilon^+,\upsilon^-,\mathbf{y})
  -\Psi_{2;\beta/ \lambda^{1/\Delta}, \gamma \lambda^{1/\Delta},1}''(\chi^+,\chi^-,\upsilon^+,\upsilon^-,\mathbf{y})
  & = 2\log\lambda.
\end{align*}
As $\Psi_2'$ is the maximum of $\Psi_2''$ over $\mathbf{y}$ and $\Psi_2$ is the maximum of $\Psi_2'$ over $(\upsilon^+,\upsilon^-)$,
the second and the third equations hold.
\end{proof}
 
 Lemma~\ref{lem:new} shows that modulo a constant term, the expressions~$\Psi_1$ and~$\Psi_2$ are
 preserved by the standard transformation on $\Delta$-regular graphs whereby an external field~$\lambda$ is
 pushed into the edge interactions.
Using this lemma, we can now draw conclusions about the maximisation of these quantities in our setting.

\begin{lemma}\label{lem:first and second moment}
  For any parameter set $(\beta,\gamma,\lambda)$ and $\Delta\geq 3$, 
  if $(\beta,\gamma,\lambda)$ lies in the non-uniqueness region of $\TreeD$, 
  then $\Psi_{2;\beta,\gamma,\lambda}(p^+,p^-)=2\Psi_{1;\beta,\gamma,\lambda}(p^+,p^-)$ and $(p^+,p^-)$ maximizes both $\Psi_{1;\beta,\gamma,\lambda}$ and $\Psi_{2;\beta,\gamma,\lambda}$.
\end{lemma}
 
\begin{proof}
First, using equalities (\ref{ssone}) and (\ref{sstwo}),
and the fact that $\E[X^2] \geq \E[X]^2$,
\begin{align*}
\Psi_{2;\beta,\gamma,\lambda}(p^+,p^-) 
&=  \frac{1}{n} \log\E_{\Gc_n} \left[ {(Z_G^{p^+,p^-})}^2 \right] + o(1)\\
&\geq    \frac{1}{n}\log {\left( \E_{\Gc_n}\big[Z_G^{p^+,p^-}\big] \right)}^2 + o(1)\\
&=  2 \frac{1}{n}\log  \E_{\Gc_n} \big[Z^{p^+,p^-}_G\big]  + o(1)\\
&= 2 \Psi_{1;\beta,\gamma,\lambda}(p^+,p^-)+o(1).
\end{align*}

But since $\Psi_{2;\beta,\gamma,\lambda}(p^+,p^-)$ and
$\Psi_{1;\beta,\gamma,\lambda}(p^+,p^-)$ don't depend upon~$n$,
we have 
$$ \Psi_{2;\beta,\gamma,\lambda}(p^+,p^-) \geq 2 \Psi_{1;\beta,\gamma,\lambda}(p^+,p^-).$$

For the other direction, we can apply Lemma~\ref{lem:GSV13:thm3}
and  Lemma~\ref{lem:new}
which together give
 
 \begin{align*}
\max_{\chi^+,\chi^-} \Psi_{2;\beta,\gamma,\lambda}(\chi^+,\chi^-)&=
\max_{\chi^+,\chi^-} \Psi_{2; \beta/ \lambda^{1/\Delta}, \gamma \lambda^{1/\Delta},1}(\chi^+,\chi^-)+2\log\lambda\\
&=  2\left(\max_{\chi^+,\chi^-} \Psi_{1; \beta/ \lambda^{1/\Delta}, \gamma \lambda^{1/\Delta},1}(\chi^+,\chi^-)+\log\lambda\right)\\
&=  2 \max_{\chi^+,\chi^-} \Psi_{1; \beta,\gamma,\lambda}(\chi^+,\chi^-).
\end{align*}

Since $(p^+,p^-)$ maximizes $\Psi_{1;\beta,\gamma,\lambda}$ when $(\beta,\gamma,\lambda)$ lies in the non-uniqueness region of $\TreeD$ (\cite[Lemma 1]{GSV12}),
  we have $\Psi_{2;\beta,\gamma,\lambda}(p^+,p^-)\leq 2\Psi_{1;\beta,\gamma,\lambda}(p^+,p^-)$.
  Therefore $\Psi_{2;\beta,\gamma,\lambda}(p^+,p^-)= 2\Psi_{1;\beta,\gamma,\lambda}(p^+,p^-)$, and $(p^+,p^-)$ also maximizes $\Psi_{2;\beta,\gamma,\lambda}$.
\end{proof}

By the same reasoning as in the proof above, 
we may apply other results from \cite{GSV13}.
Recall that $\Psi_{2;\beta,\gamma,\lambda}(\chi^+,\chi^-)=\max_{\upsilon^+,\upsilon^-} \Psi_{2;\beta,\gamma,\lambda}'(\chi^+\chi^-,\upsilon^+,\upsilon^-)$.
Let 
$$\mathcal{D}'(\chi^+,\chi^-) = \{ (\upsilon^+,\upsilon^-)\in \mathcal{D}(\chi^+,\chi^-) \mid 
\upsilon^+ \neq (\chi^+)^2 \mbox{ or }
\upsilon^- \neq (\chi^-)^2\}.$$

\begin{lemma}[\protect{\cite[Lemma 29]{GSV13}}]\label{lem:GSV13:lemma29}
For all $0\leq \chi^+ \leq 1$ and $0\leq \chi^- \leq 1$ and
all $(\upsilon^+,\upsilon^-) \in \mathcal{D}'(\chi^+,\chi^-)$,
we have

 $$ 
 \Psi_{2;\beta,\gamma,1}'(\chi^+,\chi^-,\upsilon^+,\upsilon^-)
 < 2  \max_{0\leq x \leq 1,0\leq y \leq 1}
	\Psi_{1;\beta,\gamma,1}(x,y).$$
\end{lemma}
To see that our version of Lemma~\ref{lem:GSV13:lemma29} is the same as the one in~\cite{GSV13},
see \cite[Equation (3)]{GSV13}.
So using Lemma~\ref{lem:new}
we obtain 
\begin{corollary}
\label{cor:new}
For all $0\leq \chi^+ \leq 1$ and $0\leq \chi^- \leq 1$ and
all $(\upsilon^+,\upsilon^-) \in \mathcal{D}'(\chi^+,\chi^-)$,
we have
$$ 
 \Psi_{2;\beta,\gamma,\lambda}'(\chi^+,\chi^-,\upsilon^+,\upsilon^-)
 < 2 \max_{ 0\leq x \leq 1,0\leq y \leq 1}
	\Psi_{1;\beta,\gamma,\lambda}(x,y).$$
\end{corollary}

Now we can show that Condition \ref{GSV12:condition1} holds when non-uniqueness holds.

\begin{lemma}\label{lem:condition1}
  For any parameter set $(\beta,\gamma,\lambda)$ and $\Delta\geq 3$, 
  if $(\beta,\gamma,\lambda)$ lies in the non-uniqueness region of $\TreeD$, 
  then Condition \ref{GSV12:condition1} holds.
\end{lemma}
\begin{proof}
  From
  Corollary~\ref{cor:new},
we have that for all $0\leq \chi^+ \leq 1$ and $0\leq \chi^- \leq 1$ and
all $(\upsilon^+,\upsilon^-) \in \mathcal{D}'(\chi^+,\chi^-)$,
$$ 
 \Psi_{2;\beta,\gamma,\lambda}'(\chi^+,\chi^-,\upsilon^+,\upsilon^-)
 < 2 \max_{ 0\leq x \leq 1,0\leq y \leq 1}
	\Psi_{1;\beta,\gamma,\lambda}(x,y)  
	.$$
	In particular, this holds for $\chi^+=p^+$ and $\chi^-=p^-$.
	Together with Lemma~\ref{lem:first and second moment}, it gives us
	that for all $(\upsilon^+,\upsilon^-) \in \mathcal{D}'(p^+,p^-)$,
	\begin{align*}
	  \Psi_{2;\beta,\gamma,\lambda}'(p^+,p^-,\upsilon^+,\upsilon^-)
      &< 2 \max_{ 0\leq x \leq 1,0\leq y \leq 1}
	  \Psi_{1;\beta,\gamma,\lambda}(x,y)\\
	  &= \Psi_{2;\beta,\gamma,\lambda}(p^+,p^-)\\
	  &= \max_{(\upsilon^+,\upsilon^-)\in {\mathcal D}(\chi^+,\chi^-)}\Psi_{2;\beta,\gamma,\lambda}'(p^+,p^-,\upsilon^+,\upsilon^-).
	\end{align*}
  It implies that $\Psi_{2;\beta,\gamma,\lambda}'(p^+,p^-,\upsilon^+,\upsilon^-)$ is maximized at $\upsilon^+=(p^+)^2,\upsilon^-=(p^-)^2$,
  which is Condition \ref{GSV12:condition1}.
\end{proof}

Lemma~\ref{lem:indpt-phases} follows directly.

\begin{proof}[Proof of Lemma~\ref{lem:indpt-phases}]
  By Lemma~\ref{lem:GSV12:condition1}, this lemma follows from Condition~\ref{GSV12:condition1},
  and Condition~\ref{GSV12:condition1} is established by Lemma~\ref{lem:condition1}.
\end{proof}

\subsubsection{Proof of  Equations~\eqref{ssone}, \eqref{lastone} and \eqref{sstwo}}\label{sec:asymptotics}
The proof is a rather lengthy formalisation of the following main idea: as $n$ grows large, the sums in the expressions for the first and second moments are dominated by their maximum terms and the integrality conditions for the variables may be dropped introducing only $o(1)$ (additive) error in the asymptotics of the logarithms. 

We will use the following lemma. 
Recall that for a metric space $(\mathcal{X},d)$, 
an $\epsilon$-cover of a set $\mathcal{K}\subseteq \mathcal{X}$ is a set $\mathcal{C}\subseteq \mathcal{X}$
such that every point in $\mathcal{K}$ is within distance $\epsilon$ from a point in $\mathcal{C}$. 
Note that we do \textit{not} require $\mathcal{C}\subseteq \mathcal{K}$. Recall also that $x\in
\mathcal{X}$ is a limit point of $\mathcal{K}$ if $x$ is in the closure of $\mathcal{K}\backslash \{x\}$.
\begin{lemma}\label{lem:maximalimit}
Let $h:\mathcal{R}_h\rightarrow \mathbb{R}$ be a continuous function on $\mathcal{R}_h\subseteq \mathbb{R}^k$. Let $\mathcal{R}$ be a compact subregion of $\mathcal{R}_h$.  For $n=1,2,\hdots$, let $\mathcal{C}_n\subseteq \mathcal{R}_h$ be an $\epsilon_n$-cover  of $\mathcal{R}$, where $\epsilon_n\rightarrow 0$ as $n\rightarrow \infty$. Assume further that the $\mathcal{C}_n$ are uniformly bounded and that the set of limit points of $\bigcup_{n}(\mathcal{C}_n\backslash \mathcal{R})$ is a subset of $\mathcal{R}$.
Then, as $n\rightarrow \infty$, it holds that
\[\sup_{\mathbf{x}\in\mathcal{C}_n} h(\mathbf{x})\rightarrow \max_{\mathbf{x}\in\mathcal{R}} h(\mathbf{x}).\]
\end{lemma}
\begin{proof} 
Let 
\[L_n:=\sup_{\mathbf{x}\in \mathcal{C}_n} h(\mathbf{x}),\quad L:=\max_{\mathbf{x}\in \mathcal{R}} h(\mathbf{x}).\]
Note that the maximum in the definition of $L$ is justified by compactness of the region $\mathcal{R}$ and continuity of $h$. We will prove that
\begin{equation}\label{eq:infsup}
L\leq \liminf_{n\rightarrow\infty} L_n\leq \limsup_{n\rightarrow\infty} L_n\leq L.
\end{equation}

We first prove the right-most inequality in \eqref{eq:infsup}. Assume for the sake of contradiction that $\limsup_{n\rightarrow\infty}L_n>L$. Then, there exist sequences $n_m,\mathbf{x}_{n_m}$ with $n_m\rightarrow \infty$ and $\mathbf{x}_{n_m}\in \mathcal{C}_{n_m}$ such that $h(\mathbf{x}_{n_m})>L+\epsilon$ for some $\epsilon>0$.
Since the $\mathcal{C}_n$ are uniformly bounded, $\mathbf{x}_{n_m}$ are bounded, and so they have a convergent subsequence, whose limit we denote by $\bar{\mathbf{x}}$. From the continuity of $h$, we obtain $h(\bar{\mathbf{x}})\geq L+\epsilon$. We claim that $\bar{\mathbf{x}}\in \mathcal{R}$, and thus obtain a contradiction to the choice of $L$. Indeed, if there exists $m_0$ such that $\mathbf{x}_{n_m}\in \mathcal{R}$ for all $m\geq m_0$, we have $\bar{\mathbf{x}}\in \mathcal{R}$ by compactness of $\mathcal{R}$. Otherwise, by restricting to a subsequence of $\mathbf{x}_{n_m}$ if necessary, $\bar{\mathbf{x}}$ is a limit point of $\bigcup_{n}(\mathcal{C}_n\backslash \mathcal{R})$, and hence $\bar{\mathbf{x}}\in \mathcal{R}$. 

To show the left-most inequality in \eqref{eq:infsup}, let $\mathbf{x}^{*}\in\mathcal{R}$ be such that $h(\mathbf{x}^{*})=L$. Since the  $\mathcal{C}_n$ are $\epsilon_n$-covers of $\mathcal{R}$ and $\epsilon_n\rightarrow0$, there exists a sequence $\mathbf{x}_{n}$ with $\mathbf{x}_{n}\in \mathcal{C}_{n}$ such that $\mathbf{x}_{n}\rightarrow \mathbf{x}^{*}$. Note that $L_{n}\geq h(\mathbf{x}_{n})$ and thus, from the  continuity of $h$, we obtain $\liminf_{n\rightarrow\infty} L_{n}\geq h(\mathbf{x}^{*})=L$. This finishes the proof of \eqref{eq:infsup}.
\end{proof}

We will now prove Equations~\eqref{ssone}, \eqref{lastone} and \eqref{sstwo}.
We start with the first moment $\E_{\Gc_n}\big[Z^{\chi^+,\chi^-}_G\big]$. 
Under the convention $0\log 0 = 0$,
$f_1(z,w)$ is continuous for $0\leq z \leq 1$ and $0\leq w \leq 1$.
We will write $g_1(x,z,w)$ for $g_{1;\beta,\gamma}(x,z,w)$. The function 
$g_1(x,z,w)$ is continuous
 over the region 
\begin{equation}
\label{eq:Rg}\mathcal{R}_{g_1} = \{(x,z,w) \mid
x,z,w \geq 0, x\leq z \leq 1, x\leq w \leq 1, 0\leq 1-z-w+x\}.
\end{equation}

 The following asymptotic expressions   are obtained from Stirling's approximation
 for $|\Sigma^{\chi^+,\chi^-}|$ and $\kappa^{\chi^+_n,\chi^-_n}_{x_n}$
  using  equations \eqref{one2april} and \eqref{two2april}.
\begin{equation}\label{eq:nnf}
\frac{1}{n}\log |\Sigma^{\chi^+,\chi^-}|=-{f}_1(\chi^+_n,\chi^-_n)+o(1),\quad
\frac{1}{n}\log \kappa^{\chi^+_n,\chi^-_n}_{x_n}={g}_{1}(x_n,\chi^+_n,\chi^-_n)+f_1(\chi^+_n,\chi^-_n)+o(1).
\end{equation}
Note that the sum over~$x_n$ in the expression 
\eqref{one2april} which is used in
\eqref{two2april} for $\E_{\Gc_n}\big[Z^{\chi^+,\chi^-}_G\big]$ has 
at most $n$ terms, so obviously the number of terms is at most a polynomial in~$n$.
Thus, we may approximate the sum   by its maximum term to obtain
\begin{equation}\label{eq:nfirst}
\frac{1}{n}\log \E_{\Gc_n}\big[Z^{\chi^+,\chi^-}_G\big]=
\log(\lambda)(\chi^+_n+\chi^-_n)+(\Delta-1){f}_1(\chi^+_n,\chi^-_n)+\Delta \max_{x_n\in[\chi_n,\chi_n']} {g}_{1}(x_n,\chi^+_n,\chi^-_n)+o(1).
\end{equation}
where, recall, $\chi_n=\max\{0,\chi^+_n+\chi^-_n-1\}, \ \chi'_n=\min\{\chi^{+}_n,\chi^{-}_n\}$. Moreover, as $n\rightarrow \infty$, we have 
\begin{align}
\log(\lambda)(\chi^+_n+\chi^-_n)+(\Delta-1){f}_1(\chi^+_n,\chi^-_n)&\rightarrow \log(\lambda)(\chi^++\chi^-)+(\Delta-1){f}_1(\chi^+,\chi^-),\label{eq:nf1}\\
\max_{x_n\in[\chi_n,\chi_n']}{g}_{1}(x_n,\chi^+_n,\chi^-_n)&\rightarrow \max_{x\in[\chi,\chi']} {g}_{1}(x,\chi^+,\chi^-)\label{eq:ng1},
\end{align}
where recall that $\chi=\max\{0,\chi^++\chi^--1\}, \ \chi'=\min\{\chi^{+},\chi^{-}\}$. The limit \eqref{eq:nf1}  follows from $(\chi^+_n,\chi^-_n)\rightarrow (\chi^+,\chi^-)$ and the continuity of ${f}_1$. 
To establish the
limit \eqref{eq:ng1}
we use  Lemma~\ref{lem:maximalimit} for the function $g_1$, with $\mathcal{R}_{g_1}$ as defined in (\ref{eq:Rg}), $\mathcal{R} = \{(x,\chi^+,\chi^-) \mid \chi \leq x \leq \chi'\}$
and $\mathcal{C}_n = \{(x_n,\chi_n^+,\chi_n^-) \mid \chi_n \leq x_n \leq \chi'_n,\, nx_n\in\mathbb{Z}\}$. The $\mathcal{C}_n$ are $\frac{10}{n}$-covers of $\mathcal{R}$ (with room to spare for the constant 10) under the Euclidean metric on $\mathbb{R}^3$. It is also immediate to verify that the set of limit points of $\bigcup_n  (\mathcal{C}_n\backslash \mathcal{R})$ is a subset of $\mathcal{R}$ as a consequence of $(\chi^+_n,\chi^-_n,\chi_n,\chi_n')\rightarrow (\chi^+,\chi^-,\chi,\chi')$.  
 
From the definition \eqref{eq:firstmoment}, 
\begin{equation}\label{eq:psi1}
\Psi_{1;\beta,\gamma,\lambda}(\chi^+,\chi^-)= \log(\lambda)(\chi^++\chi^-)+
(\Delta-1){f}_1(\chi^+,\chi^-)+\Delta \max_{\chi\leq x\leq \chi'} {g}_{1}(x,\chi^+,\chi^-),
\end{equation}  Combining 
this with \eqref{eq:nfirst}, \eqref{eq:nf1} and \eqref{eq:ng1}  yields \eqref{ssone}, as wanted.

We next turn to the second moment $\E_{\Gc_n}\big[(Z^{\chi^+,\chi^-}_G)^2\big]$. The reasoning is almost identical, though carrying out the arguments is more cumbersome due to the larger number of variables. 
Let 
$$\mathcal{R}_{f_2} = 
\{(z,w,u,v) \mid z\geq u \geq 0,
1-2z+u\geq 0,
w\geq v\geq 0, 
1-2w+v\geq 0 \}.$$
The function $f_2(z,w,u,v)$ is continuous over~$\mathcal{R}_{f_2}$. We will also write $g_2(\mathbf{y})$ instead of $g_{2;\beta,\gamma}(\mathbf{y})$. Note that $g_2(\mathbf{y})$ is continuous over
$\mathcal{R}_{g_2}=\{\mathbf{y} \mid \forall i,j\in[4], y_{ij} \geq 0\}$.
By analogy to \eqref{eq:LRlimits}, define 
\begin{equation*} 
\begin{gathered}
l_1=u,\quad l_2=l_3=z-u,\quad l_4=1-2z+u,\\
r_1=v,\quad r_2=r_3=w-v,\quad r_4=1-2w+v.
\end{gathered}
\end{equation*}  
For $(z,w,u,v)\in \mathcal{R}_{f_2}$,
 define
$$\mathcal{T}(z,w,u,v) = \big\{ \mathbf{y}=\{y_{ij}\}_{i,j\in[4]} \mid 
\forall i,j\in[4],
y_{ij} \geq 0,\sum_j y_{ij}=l_ i,
\sum_i y_{ij} =r_j  \big\}.$$  
Let $\mathcal{R} = \mathcal{T}(\chi^+,\chi^-,\upsilon^+,\upsilon^-)$
and let $\mathcal{C}_n = \mathcal{T}(\chi^+_n,\chi^-_n,\upsilon^+_n,\upsilon^-_n)$.
 
By Stirling's approximation, the functions ${f}_2$ and ${g}_2$ capture the asymptotics of 
$\Sigma^{\chi^+,\chi^-}_{\upsilon^+,\upsilon^-}$ and
 $\xi^{\chi^+_n,\chi^-_n}_{\upsilon^+_n,\upsilon^-_n,\mathbf{y}_n}$ (see \eqref{eq:sigma2def} and \eqref{eq:defxi}):

\begin{equation}\label{eq:nns}
\frac{1}{n}\log |\Sigma^{\chi^+,\chi^-}_{\upsilon^+,\upsilon^-}|=-{f}_2(\chi^+_n,\chi^-_n,\upsilon^+_n,\upsilon^-_n)+o(1),\quad
\frac{1}{n}\log \xi^{\chi^+_n,\chi^-_n}_{\upsilon^+_n,\upsilon^-_n,\mathbf{y}_n}=
{g}_{2}(\mathbf{y}_n)+f_2(\chi^+_n,\chi^-_n,\upsilon^+_n,\upsilon^-_n)+o(1).
\end{equation}
We are now set to compute the asymptotics of $\E_{\Gc_n}\big[Y^{\chi^+,\chi^-}_{\upsilon^+,\upsilon^-}\big]$, for which we have  derived an expression in \eqref{longeq}. Approximating the sums in the expression by their maximum terms only introduces $o(1)$ error, so we obtain
\begin{align}\label{eq:nY}
\frac{1}{n}\log\E_{\Gc_n}\big[Y^{\chi^+,\chi^-}_{\upsilon^+,\upsilon^-}\big]&=2\log(\lambda)(\chi^+_n+\chi^-_n)+(\Delta-1){f}_2(\chi^+_n,\chi^-_n, \upsilon^+_n,\upsilon^-_n)+\Delta \max_{\mathbf{y}_n\in \mathcal{C}_n} {g}_{2}(\mathbf{y}_n)+o(1).
\end{align}

 Taking limits as $n\rightarrow \infty$, we will now  derive analogues of \eqref{eq:nf1} and \eqref{eq:ng1}:
\begin{align}
2\log(\lambda)(\chi^+_n+\chi^-_n)+(\Delta-1){f}_2(\chi^+_n,\chi^-_n, \upsilon^+_n,\upsilon^-_n)&\rightarrow 2\log(\lambda)(\chi^++\chi^-)+(\Delta-1){f}_2(\chi^+,\chi^-, \upsilon^+,\upsilon^-),\label{eq:nf2}\\
\max_{\mathbf{y}_n\in  \mathcal{C}_n} {g}_{2}(\mathbf{y}_n)&\rightarrow \max_{\mathbf{y}\in \mathcal{R}} {g}_{2}(\mathbf{y}).\label{eq:ng2}
\end{align}
Equation~(\ref{eq:nf2}) follows from the continuity of~$f_2$. 
Equation~(\ref{eq:ng2}) follows by Lemma~\ref{lem:maximalimit}.
Combining \eqref{eq:nY}, \eqref{eq:nf2} and \eqref{eq:ng2} gives 
\begin{equation*}
\frac{1}{n}\log\E_{\Gc_n}\big[Y^{\chi^+,\chi^-}_{\upsilon^+,\upsilon^-}\big]=2\log(\lambda)(\chi^++\chi^-)+
(\Delta-1){f}_2(\chi^+,\chi^-, \upsilon^+,\upsilon^-)+\Delta \max_{\mathbf{y}\in \mathcal{T}(\chi^+,\chi^-, \upsilon^+,\upsilon^-)} {g}_{2}(\mathbf{y})+o(1).
\end{equation*}
The right-hand side is  $\Psi'_{2;\beta,\gamma,\lambda}(\chi^+,\chi^-,\upsilon^+,\upsilon^-)+o(1)$,  so we have 
established~\eqref{lastone}. 
 
To obtain  
\eqref{sstwo}, the sum \eqref{eq:sum2nd} yields
\begin{equation}\label{eq:nsstwo}
\begin{aligned}
\frac{1}{n}\log\E_{\Gc_n}\big[(Z_G^{\chi^+,\chi^-})^2\big]&=2\log(\lambda)(\chi^+_n+\chi^-_n)\\
&+\max_{(\upsilon^+_n,\upsilon^-_n,\mathbf{y}_n)\in \mathcal{O}_n}\{ (\Delta-1){f}_2(\chi^+_n,\chi^-_n, \upsilon^+_n,\upsilon^-_n)+\Delta {g}_{2}(\mathbf{y}_n)\}+o(1),
\end{aligned}
\end{equation}
where 
\begin{equation*}
\begin{aligned}
\mathcal{D}_n &:= \{ (\upsilon^+_n,\upsilon^-_n) \mid
\max\{0,2 \chi^+_n-1\} \leq \upsilon^+_n \leq \chi^+_n,\ 
\max\{0,2 \chi^-_n-1\}\leq \upsilon^-_n \leq \chi^-_n
\},\\ 
\mathcal{O}_n&:=\{(\upsilon^+_n,\upsilon^-_n,\mathbf{y}_n)\mid (\upsilon^+_n,\upsilon^-_n)\in \mathcal{D}_n,\ \mathbf{y}_n\in\mathcal{C}_n\}.
\end{aligned}
\end{equation*} Yet another application of Lemma~\ref{lem:maximalimit} yields
\begin{equation}\label{eq:nf2g2}
\begin{aligned}
\max_{(\upsilon^+_n,\upsilon^-_n,\mathbf{y}_n)\in \mathcal{O}_n}&\{ (\Delta-1){f}_2(\chi^+_n,\chi^-_n, \upsilon^+_n,\upsilon^-_n)+ \Delta {g}_{2}(\mathbf{y}_n)\}\\
&\rightarrow \max_{(\upsilon^+,\upsilon^-,\mathbf{y})\in \mathcal{O}(\chi^+,\chi^-)}\{ (\Delta-1){f}_2(\chi^+,\chi^-, \upsilon^+,\upsilon^-)+\Delta {g}_{2}(\mathbf{y})\},
\end{aligned}
\end{equation}
where $\mathcal{O}(\chi^+,\chi^-):=\{(\upsilon^+,\upsilon^-,\mathbf{y})\mid (\upsilon^+,\upsilon^-)\in\mathcal{D},\ \mathbf{y}\in \mathcal{T}(1,\chi^+,\chi^-, \upsilon^+,\upsilon^-)\}$. Combining \eqref{eq:nsstwo} and \eqref{eq:nf2g2} yields \eqref{sstwo}, as desired.

\end{document}